\RequirePackage{tikz}
\documentclass[sn-mathphys-num]{sn-jnl}

\usepackage{fixmath}
\usepackage{bm}
\usepackage{amsbsy}
\usepackage{amssymb}
\usepackage{amsthm}
\usepackage{amsmath}
\usepackage{mathtools}
\usepackage{verbatim}
\usepackage{multirow}
\usepackage{array}
\usepackage{graphicx}
\usepackage{tikz}
\usetikzlibrary{positioning}
\usepackage{color}
\usepackage{xcolor}

\usepackage{thm-restate}

\newcommand{\size}[1]{\ensuremath{|#1|}}
\newcommand{\ceil}[1]{\ensuremath{\lceil#1\rceil}}
\newcommand{\floor}[1]{\ensuremath{\lfloor#1\rfloor}}

\newcommand{\lrA}[1]{\ensuremath{\left(#1\right)}}

\newcommand{\lrC}[1]{\ensuremath{\left\{#1\right\}}}

\let\epsilon=\varepsilon

\def\MST{\mbox{MST}}
\def\C{\mathcal{C}}

\def\P{\mathcal{P}}

\def\U{\mathcal{U}}
\def\E{\mathcal{E}}
\def\F{\mathcal{F}}
\def\H{\mathcal{H}}
\def\G{\mathcal{G}}
\def\V{\mathcal{V}}

\def\MST{\mbox{MST}}
\def\C{\mathcal{C}}

\def\P{\mathcal{P}}

\usepackage{algorithm}
\usepackage{algorithmicx}
\usepackage{algpseudocode}

\newcommand{\EE}[1]{\ensuremath{\mathbb{E}[#1]}}

\newtheorem{theorem}{Theorem}
\newtheorem{lemma}[theorem]{Lemma}
\newtheorem{corollary}[theorem]{Corollary}

\newtheorem{assumption}[theorem]{Assumption}

\begin{document}

\title[TTP: Improved Algorithms Based on Cycle Packing]{The Traveling Tournament Problem: Improved Algorithms Based on Cycle Packing}


\author{\fnm{Jingyang} \sur{Zhao} }\email{jingyangzhao1020@gmail.com}

\author{\fnm{Mingyu} \sur{Xiao}}\email{myxiao@uestc.edu.cn}
\equalcont{Corresponding author}

\author{\fnm{Chao} \sur{Xu} }\email{the.chao.xu@gmail.com}

\affil{\orgdiv{School of Computer Science and Engineering}, \orgname{University of Electronic Science and Technology of China}, \orgaddress{\street{2006 Xiyuan Ave}, \city{Chengdu}, \postcode{610054}, \state{Sichuan}, \country{China}}}

\abstract{
The Traveling Tournament Problem (TTP) is a well-known benchmark problem in the field of tournament timetabling, which asks us to design a double round-robin schedule such that each pair of teams plays one game in each other's home venue, minimizing the total distance traveled by all $n$ teams ($n$ is even).
TTP-$k$ is the problem with one more constraint that each team can have at most $k$-consecutive home games or away games.
In this paper, we investigate schedules for TTP-$k$ and analyze the approximation ratio of the solutions.
Most previous schedules were constructed based on a Hamiltonian cycle of the graph.
We will propose a novel construction based on a $k$-cycle packing. Then, combining our $k$-cycle packing schedule with the Hamiltonian cycle schedule, we obtain improved approximation ratios for TTP-$k$ with deep analysis.
The case where $k=3$, TTP-3, is one of the most investigated cases. We improve the approximation ratio of TTP-3 from $(1.667+\epsilon)$ to $(1.598+\epsilon)$, for any $\epsilon>0$.
For TTP-$4$, we improve the approximation ratio from $(1.750+\epsilon)$ to $(1.700+\epsilon)$.
By a refined analysis of the Hamiltonian cycle construction, we also improve the approximation ratio of TTP-$k$ from $(\frac{5k-7}{2k}+\epsilon)$ to $(\frac{5k^2-4k+3}{2k(k+1)}+\epsilon)$ for any constant $k\geq 5$.
Our methods can be extended to solve a variant called LDTTP-$k$ (TTP-$k$ where all teams are allocated on a straight line). We show that the $k$-cycle packing construction can achieve an approximation ratio of $(\frac{3k-3}{2k-1}+\epsilon)$, which improves the approximation ratio of LDTTP-3 from $4/3$ to $(6/5+\epsilon)$.
}




\maketitle
\section{Introduction}
In the field of the tournament schedule, the traveling tournament problem (TTP) is a widely known benchmark problem that was first systematically introduced in~\cite{easton2001traveling}.
This problem aims to find a double round-robin tournament satisfying some constraints, minimizing the total distance traveled by all participant teams.
In a double round-robin tournament of $n$ teams, each team will play 2 games against each of the other $n-1$ teams, one at its home venue and one at its opponent's home venue.
Additionally, each team should play one game a day, all games need to be scheduled on $2(n-1)$-consecutive days, and so there are exactly $n/2$ games on each day.
According to the definition, we know that $n$ is always even.
For TTP, we have the following two basic constraints or assumptions on the double round-robin tournament.

\medskip
\noindent
\textbf{Traveling Tournament Problem (TTP)}
\begin{itemize}
\item \emph{No-repeat}: Two teams cannot play against each other on two consecutive days.
\item \emph{Direct-traveling}: Before the first game, all teams are at home, and they will return home after the last game. Furthermore, a team travels directly from its game venue on $i$th day to its game venue on $(i+1)$th day.
\end{itemize}

TTP-$k$, a well-known variant of TTP, is to add the following constraint on the maximum number of consecutive home games and away games.

\begin{itemize}
\item \emph{Bounded-by-$k$}: Each team can have at most $k$-consecutive home games or away games.
\end{itemize}

The smaller the value of $k$, the more frequently a team has to return home. By contrast, if $k$ is very large, say $k=n-1$, then this constraint loses meaning, and TTP-$k$ becomes TTP where a team can schedule their travel distance as short as that in the traveling salesman problem (TSP).

A weight function $w$ on the complete undirected graph is called a \emph{metric} if it satisfies the symmetry and triangle inequality
properties: $w(a,b) = w(b,a)$ and $w(a,c) \leq w(a,b) + w(b,c)$ for all $a,b,c\in V$.

The input of TTP-$k$ contains a complete graph where each vertex is a team, and the weight between two vertices $u$ and $v$ is the distance from the home of team $u$ to the home of team $v$. In this paper, we only consider the case when the weight function $w$ is a metric.

\subsection{Related Work}
TTP and TTP-$k$ are difficult optimization problems.
The NP-hardness of TTP and TTP-$k$ with $k\geq 3$ has been established \cite{bhattacharyya2016complexity, thielen2011complexity, DBLP:journals/corr/abs-2110-02300}.
Although the hardness of TTP-2 is still not formally proved, it is believed that TTP-2 is also hard since it is not easy to construct a feasible solution to it.
In the literature there is a large number of contributions on approximation algorithms~\cite{thielen2012approximation,DBLP:conf/mfcs/XiaoK16,DBLP:conf/ijcai/ZhaoX21,DBLP:conf/cocoon/ZhaoX21,DBLP:conf/atmos/ChatterjeeR21,miyashiro2012approximation,yamaguchi2009improved,imahori2010approximation,westphal2014,hoshino2013approximation,imahori20211+} and heuristic algorithms~\cite{easton2003solving,lim2006simulated,anagnostopoulos2006simulated,di2007composite,goerigk2014solving,DBLP:journals/anor/GoerigkW16}.

For heuristic algorithms, most known works are concerned with the case of $k=3$.
Since the search spaces of TTP and TTP-$k$ are very large, many instances with more than 10 teams in the online benchmark \cite{trick2007challenge,DBLP:journals/eor/BulckGSG20}
have not been completely solved even by using high-performance machines.

In terms of approximation algorithms, most results are based on the assumption that the distance holds the symmetry and triangle inequality properties. This is natural and practical in the sports schedule. For $k=2$, one significant contribution to TTP-2 was done by Thielen and Westphal~\cite{thielen2012approximation}. They proposed an $(1+16/n)$-approximation algorithm for $n/2$ being even and a $(3/2+O(1/n))$-approximation algorithm for $n/2$ being odd. Approximation ratios for these two cases have been improved to $(1+3/n)$ and $(1+12/n)$, respectively~\cite{DBLP:conf/cocoon/ZhaoX21,DBLP:conf/ijcai/ZhaoX21}.
For $k=3$, the first approximation algorithm, proposed by Miyashiro \emph{et al.}~\cite{miyashiro2012approximation}, admits a $2+O(1/n)$ approximation ratio by using the Modified Circle Method. Then, the approximation ratio was improved to $5/3+O(1/n)$ by Yamaguchi \emph{et al.}~\cite{yamaguchi2009improved}. Their approximation algorithm also works for $3<k=o(n)$. For $k=4$, the ratio is $1.750+O(1/n)$, and for $4<k=o(n)$, the ratio is $(5k-7)/(2k)+O(k/n)$~\cite{yamaguchi2009improved}.
For $k=\Theta(n)$, Westphal and Noparlik~\cite{westphal2014} proposed a $5.875$-approximation algorithm for any choice of $k\geq 4$ and $n\geq 6$, and Imahori \emph{et al.}~\cite{imahori2010approximation} gave a $2.75$-approximation algorithm for $k=n-1$ where they also proved that the ratio could be improved to $2.25$ if the optimal TSP is given. In both of their algorithms for $k=\Theta(n)$, the approximation ratios can be slightly improved if they use the randomized $(3/2-\epsilon_0)$-approximation algorithm for TSP \cite{DBLP:conf/stoc/KarlinKG21}.
The current best approximation ratio for $k=3$ is still $5/3+O(1/n)$~\cite{yamaguchi2009improved}.

We refer the readers to \cite{DBLP:journals/eor/BulckGSG20} for more variants of TTP (including the traveling tournament problem with predefined venues, the time-relaxed traveling tournament problem, etc.).

\subsection{Our Contributions}
In this paper, we consider approximation algorithms for TTP-$k$ with constant $k\geq3$.
Our contributions are summarized as follows.

Firstly, we deeply analyze the structural properties, which leads us to strengthen some of the current-known lower bounds and also to propose several new lower bounds.

Secondly, we consider the constructions of feasible schedules for TTP-$k$.
Our algorithm will use two constructions.
The first construction is based on a Hamiltonian cycle, which is a known method. The Hamiltonian cycle used in the literature is usually generated by the Christofides-Serdyukov algorithm~\cite{christofides1976worst,serdyukov1978some}, while we propose a new Hamiltonian cycle that uses the minimum weight perfect matching.
The second construction is newly proposed by us, which is based on a $k$-cycle packing.
To get our claimed approximation ratios, we will make a trade-off between the two constructions and do careful analysis.

Thirdly, based on careful analysis, we improve the approximation ratio from $(5/3+\epsilon)$ to $(139/87+\epsilon)$ for TTP-3, and from $(7/4+\epsilon)$ to $(17/10+\epsilon)$ for TTP-4.

Fourthly, we consider TTP-$k$ with $k\geq 5$. By a refined analysis of the Hamiltonian construction, we improve the approximation ratio from $(\frac{5k-7}{2k}+\epsilon)$ to $(\frac{5k^2-4k+3}{2k(k+1)}+\epsilon)$.

Last, we show that the construction based on the $k$-cycle packing can also be used to improve Linear Distance TTP-$k$ (LDTTP-$k$), which is a special case of TTP-$k$ where all teams are on a line. We prove an approximation ratio of $(\frac{3k-3}{2k-1}+\epsilon)$ for this problem, which improves the approximation ratio from $4/3$~\cite{DBLP:journals/jair/HoshinoK12} to $(6/5+\epsilon)$ for the case of $k=3$. 

\subsection{Paper Organization}
The remaining parts of the paper are organized as follows.
In Section~\ref{sec_pre}, we introduce basic notations.
In Section \ref{A}, we propose several new lower bounds. In Section \ref{B}, we introduce the construction methods for the schedules. Specifically, in Section \ref{B.1}, we  modify the known construction based on a Hamiltonian cycle. In Section \ref{B.2}, we propose a new construction based on a $k$-cycle/path packing.
Although these two constructions cannot make any improvement separately, in Section \ref{C}, we show that together they can guarantee an improved approximation ratio for TTP-3 and TTP-4.
In Section \ref{C.1}, we analyze the approximation ratio of TTP-3.
In Section \ref{C.2}, we improve the approximation ratio for TTP-4. In Section \ref{C.3}, we give a refined analysis of the Hamiltonian construction and improve the approximation ratio for TTP-$k$ with $k\geq 5$. In Section~\ref{D}, we analyze the approximation ratio of our algorithm for LDTTP-$k$ with $k\geq3$. Finally, we make the concluding remarks in Section \ref{E}.

Partial results of this paper were presented at the 47th International Symposium on Mathematical Foundations of Computer Science (MFCS 2022)~\cite{zhao2022improved}.

\section{Preliminaries}\label{sec_pre}
We always use $n$ to denote the total number of teams in the problem, where $n$ is even. The set of $n$ teams is denoted by $V=\{t_1,t_2,\ldots,t_n\}$.
We use $G=(V, E)$ to denote the complete graph on the $n$ vertices representing the $n$ teams. There is a positive weight function $w: E\to \mathbb{R}_{\geq0}$ on the edges of $G$. We often write $w(u,v)$ to denote the weight of the edge $(u,v)$, which is the distance between the home of team $u$ and the home of team $v$.
For any weight function $w:X\to \mathbb{R}_{\geq0}$, we extend it to subsets of $X$. Define $w(Y) = \sum_{x\in Y} w(x)$ for $Y\subseteq X$. For the sake of analysis, we will always consider that $k$ is a constant not smaller than 3. Then, in the following, we also assume that $n\geq 8k^2$. Otherwise, since $k$ is a constant, we can solve TTP-$k$ in constant time.

The following notations work for any graph. We consider w.l.o.g. a graph $G=(V,E)$. For any $V'\subseteq V$, the induced graph of $G$ by $V'$ is denoted as $G[V']$.
The weight of a minimum weight spanning tree in $G$ is denoted by $\mbox{MST}(G)$.
We use $\delta_G(u)$ to denote the set of edges incident on a vertex $u$ in $G$.
When the graph is an edge-weighted graph,
we use $\deg_G(u)$ to denote the weighted degree of a vertex, i.e., sum of the weights of all edges incident on $u$. We also let $\Delta_G$ be the sum of weighted degrees in $G$, i.e., $$\Delta_G=\sum_{u\in V}\deg_G(u)=2w(E).$$

Let $m=2\floor{\frac{n}{4k}}$. Let $S$ be the set of $2mk$ vertices of $V$ such that the sum of the weighted degrees of vertices in $S$ is maximized. The set $S$ can be found in $O(n^2)$ time by greedily selecting the vertex with maximum weighted degree.
We always use $H=G[S]$ to denote the complete graph induced by $S$ in $G$.
We have
\begin{eqnarray}\label{degree}
\sum_{v\in S}\deg_G(v)\geq \frac{2mk}{n}\sum_{v\in V}\deg_G(v)=\frac{2mk}{n}\Delta_G.
\end{eqnarray}

We let $\overline{S}=V\setminus S$, $n_S=2mk$ and $n_{\overline{S}}=n-n_S=n-4k\floor{\frac{n}{4k}}$. Then, there are $n_S$ vertices in $S$. Note that $n_{\overline{S}}<4k$. Since $n\geq 8k^2$, we have
\begin{eqnarray}\label{n_S}
\frac{1}{n_S}=\frac{1}{n-n_{\overline{S}}}<\frac{1}{n-4k}<\frac{2}{n}.
\end{eqnarray}
By (\ref{degree}), we have
\begin{eqnarray}\label{degree+}
\sum_{v\in \overline{S}}\deg_G(v)\leq \frac{n_{\overline{S}}}{n}\sum_{v\in V}\deg_G(v)\leq \frac{4k}{n}\Delta_G.
\end{eqnarray}

Two subgraphs or sets of edges are \emph{vertex-disjoint} if they do not share a common vertex.
A simple path on $k$ different vertices $\{v_1,v_2,\dots,v_k\}$ is called a \emph{$k$-path} and denoted by $(v_1,v_2,\dots,v_k)$.
A simple cycle on $k$ different vertices $\{v_1,v_2,\dots,v_k\}$ is called a \emph{$k$-cycle} and denoted by $(v_1,v_2,\dots,v_k,v_1)$.
An $l$-cycle is called a \emph{mod $k$-cycle} if $l\bmod k=0$. Hence, a $k$-cycle is also a mod $k$-cycle.
A \emph{$k$-path packing} (resp., \emph{(mod) $k$-cycle packing}) in a graph is a set of edges such that each component is a $k$-path (resp., (mod) $k$-cycle) and all vertices in the graph are covered. A 2-path packing is also called a \emph{perfect matching}.

We can obtain a $k$-cycle packing from a $k$-path packing by connecting the first and last vertices of each $k$-path in the packing. We refer to this operation as \emph{completing the $k$-path packing}.
Since $n_S\bmod k=0$ and $n_S\bmod 2=0$, there must exist a $k$-cycle packing, $k$-path packing, and a perfect matching in $H$.
Usually, we let $\P_k$ denote a $k$-path packing in $H$, $\C_k$ denote a $k$-cycle packing in $H$, $\C_{\bmod k}$ denote a mod $k$-cycle packing in $H$, and $M_H$ denote a perfect matching in $H$.

Recall that $m=\frac{n_S}{2k}$. For a $k$-cycle packing $\C_k$ in graph $H$, there are exactly $2m$ $k$-cycles in $\C_k$.
We assume that the $2m$ $k$-cycles are $\{C_1,\dots,C_{2m}\}$ and $C_i=(t_{k(i-1)+1}, \dots, t_{k(i-1)+k},t_{k(i-1)+1})$ for each $1\leq i \leq 2m$.
In our algorithm, we may consider the set of teams in each cycle $C_i$ together.
For the sake of presentation, we take each cycle $C_i$ as a \emph{cycle-team} $u_i$, which is the set of teams in $C_i$, i.e., $u_i=\{t_{k(i-1)+1}, \dots, t_{k(i-1)+k}\}$.
Each cycle-team contains exactly $k$ (normal) teams.
Furthermore, we also call $U_i=u_{2i-1}\cup u_{2i}$ $(1\leq i \leq m)$ a \emph{super-team}.
Each supper-team contains exactly $2k$ (normal) teams.
Let $\U=\{U_1,\dots,U_m\}$ denote the set of all super-teams.

According to the $2m$ cycle-teams $\V=\{u_1,\dots,u_{2m}\}$, we define a compete graph $\G=(\V,\E)$, where we let $w(u_i,u_j)=\sum_{t_{i'}\in u_i\&t_{j'}\in u_j}w(t_{i'},t_{j'})$ for $i\neq j$ and $w(u_i,u_i)=0$.
According to the $m$ super-teams $\U=\{U_1,\dots,U_m\}$, we similarly define another compete graph $\H=(\U,\F)$, where we let $w(U_i,U_j)=\sum_{u_{i'}\in U_i\&u_{j'}\in U_j}w(u_{i'},u_{j'})=\sum_{t_{i'}\in U_i\&t_{j'}\in U_j}w(t_{i'},t_{j'})$ for $i\neq j$ and $w(U_i,U_i)=0$.
We also define $w(u_i)=\sum_{t_{i'},t_{j'}\in C_i}w(t_{i'},t_{j'})$ and $w(U_i)=\sum_{t_{i'},t_{j'}\in U_i}w(t_{i'},t_{j'})$. Note that $w(X,Y)$ can be seen as the sum of the weights of edges in the cut $(X,Y)$ of graph $G[X\cup Y]$ and $w(X)$ can be seen as the sum of the weights of edges in the graph $G[X]$.
Recall that $\Delta_H$ is the sum of weighted degrees in $H$. We use $E_H$ to denote the set of edges in graph $H$.
We can get
\begin{eqnarray} \label{eqn_GH}
\frac{1}{2}\Delta_H=w(E_H)=w(\V)+w(\E).
\end{eqnarray}

\section{The Independent Lower Bounds}\label{A}
For TTP and TTP-$k$, a well-known method to obtain the lower bounds is to use an independent relaxation~\cite{campbell1976minimum,easton2001traveling}.
The basic idea is to obtain a lower bound on the traveling distance of each team independently without considering the feasibility of other teams and then sum them together.
Although there exist many independent lower bounds for TTP-$k$ \cite{miyashiro2012approximation,yamaguchi2009improved,westphal2014}, we cannot use them directly to get our results. We are interested in a stronger bound.
Before we make some observations on the independent lower bounds, we first need to explore some properties of TTP-$k$.

For TTP-$k$, each team needs to visit each other team's home once in the tournament.
A \emph{road trip} of a team $v$ is a simple cycle starting and ending at $v$.
An \emph{itinerary} of team $v\in V$ on $G$ is a connected subgraph of $G$ that consists of road trips with each simple cycle of length at most $k+1$, and each vertex in $V$ other than $v$ has degree $2$.
An itinerary is \emph{optimal} for a team if it is the itinerary of minimum weight.
Considering an optimal itinerary $I_v$ for team $v$, we use $\psi_v=w(I_v)$ to denote the weight of the optimal itinerary for team $v$. Then, $\psi=\sum_{v\in V}\psi_v$ is a simple independent lower bound for the minimum weight solution of TTP-$k$. Note that this lower bound was also used in the experiment~\cite{easton2001traveling}. However, it is NP-hard to compute $\psi_v$ since it involves solving the $k$-tour problem and this problem is NP-hard for any fixed $k\geq3$~\cite{DBLP:conf/stoc/AsanoKTT97}. Hence, we want to find an alternative lower bound for each team's optimal itinerary.

When $n\bmod {k}=0$, we can prove that there always exists an optimal itinerary with no 2-cycles.

\begin{lemma}\label{2-cycle}
For TTP-$k$ with $k\geq3$, when $n\bmod {k}=0$, there exists an optimal itinerary with no 2-cycles for each team.
\end{lemma}
\begin{proof}
Consider an optimal itinerary $I_v$ of team $v$ with the minimum number of 2-cycles. Assume $I_v$ contains a 2-cycle. Since all cycles share the vertex $v$,
by the triangle inequality, we can get a $(k'+1)$-cycle by shortcutting one 2-cycle and one $k'$-cycle, and the $(k'+1)$-cycle has a weight of no greater than the sum of the weights of the two cycles.
This shows that $I_v$ contains no $k'$-cycles ($3\leq k'\leq k$) and at most one $2$-cycle if it contains one 2-cycle.
Thus, $I_v$ contains only one 2-cycle, and the rest cycles are all $(k+1)$-cycles. We get $n\bmod k=2$, a contradiction to $n\bmod k=0$.
\end{proof}

The number of vertices in graph $H$ is $n_S=2mk$, which is divisible by $k$. We may consider an optimal itinerary of a team $v\in S$ on graph $H$, denoted by $I^H_v$. Recall that $I_v$ is an optimal itinerary of the team $v$ on graph $G$.
By the triangle inequality, it is easy to see that $w(I^H_v)\leq w(I_v)$.
Next, we consider several lower bounds on $w(I^H_v)$, which can also be used as lower bounds of $w(I_v)$.

\subsection{Bounds on the Optimal Itinerary}
In this subsection, we fix a single team $v\in S$ and consider the optimal itinerary $I^H_v$ on graph $H$ for this team.
We will give a more refined analysis than that in the previous paper~\cite{yamaguchi2009improved}.

By Lemma~\ref{2-cycle}, there exists an optimal itinerary $I^H_v$ with no 2-cycles, and it consists of a set of $(k+1)$-cycles and a set of $k'$-cycles ($3\leq k'\leq k$). Hence we will partition $I^H_v$ into two sets: 
\begin{enumerate}
\item $I^H_{v,1}$ consists of all $(k+1)$-cycles; and
\item $I^H_{v,2}$ consists of the other $k'$-cycles.
\end{enumerate}
Define $0\leq \gamma_v\leq 1$, such that $$w(I^H_{v,1}) = \gamma_v w(I^H_v), \quad\mbox{and}\quad w(I^H_{v,2}) = (1-\gamma_v) w(I^H_v). $$
Hence $\gamma_v$ measures the proportion of weights of the $(k+1)$-cycles compared to the entire itinerary.
The edges in graph $H$ incident to $v$ in $I^H_v$ are called \emph{home-edges}, which is the same as $\delta_H(v)\cap I^H_v$. Let $\alpha_v$ and $\beta_v$ be the proportion of weights of the home-edges in $I^H_{v,1}$ and $I^H_{v,2}$, respectively. Namely,
$$\alpha_v w(I^H_{v,1}) = w(I^H_{v,1}\cap \delta_H(v))$$ is the weight of all home-edges in $I^H_{v,1}$, and $$\beta_v w(I^H_{v,2}) = w(I^H_{v,2} \cap \delta_H(v))$$ is the weight of all home-edges in $I^H_{v,2}$.

Now, we are ready to give bounds for the optimal itinerary $I^H_v$.

\begin{lemma}\label{lb-tsp}
Let $C_H$ be a minimum weight Hamiltonian cycle in graph $H$. It holds that $w(I^H_v)\geq w(C_H)$.
\end{lemma}
\begin{proof}
According to the definition of the itinerary, we know that each vertex in $I^H_v$ has an even degree. Then, we can get an Eulerian tour in $I^H_v$ and obtain a Hamiltonian cycle in graph $H$ by shortcutting $I^H_v$. Hence
$w(I^H_v)\geq w(C_H)$.
\end{proof}

We also recall the following result.
\begin{lemma}[\cite{christofides1976worst,serdyukov1978some}]\label{lb-chris}
For graph $H$, let $C_H$ be a minimum weight Hamiltonian cycle and $C'_H$ be a Hamiltonian cycle obtained by the Christofides-Serdyukov algorithm. It holds that $w(C'_H)\leq \mbox{MST}(H)+\frac{1}{2}w(C_H)$.
\end{lemma}

For ease of proofs in the rest of the section, we also define $S_{v,1}$ and $S_{v,2}$ that partition the vertices in $S\setminus\{v\}$, where $S_{v,1}$ consists of all vertices in $I^H_{v,1}$ except $v$, and $S_{v,2}$ consists of all vertices in $I^H_{v,2}$ except $v$. Our results are mostly simple counting arguments.

\begin{lemma}\label{lb-delta-}
It holds that $(\frac{k-1}{2}-\frac{1}{2}\gamma_v+\alpha_v\gamma_v)w(I^H_v)\geq(\frac{k-2}{2}+\alpha_v)\gamma_v w(I^H_v)+(\frac{k-3}{2}+\beta_v)(1-\gamma_v) w(I^H_v)\geq \deg_H(v)$.
\end{lemma}
\begin{proof}
By the definition of $\beta_v$, we have $1\geq\beta_v$. Hence, we have
\begin{align*}
&\lrA{\frac{k-1}{2}-\frac{1}{2}\gamma_v+\alpha_v\gamma_v}w(I^H_v)\\
&=\lrA{\frac{k-2}{2}+\alpha_v}\gamma_v w(I^H_v)+\lrA{\frac{k-3}{2}+1}(1-\gamma_v) w(I^H_v)\\
&\geq\lrA{\frac{k-2}{2}+\alpha_v}\gamma_v w(I^H_v)+\lrA{\frac{k-3}{2}+\beta_v}(1-\gamma_v) w(I^H_v).
\end{align*}
The left inequality holds.
Now we consider the right inequality.

First, we can see $\deg_H(v)=\sum_{u\in S_{v,1}} w(u,v)+\sum_{u\in S_{v,2}}w(u,v)$.

Next, recall that $\alpha_v\gamma_vw(I^H_v) = w(I^H_{v,1}\cap \delta_H(v))$ is the total weight of all home-edges in $(k+1)$-cycles.
For any $(k+1)$-cycle $C=(v,v_1,\dots,v_k,v)$ in $I^H_{v,1}$, note that the edges $(v,v_1)$ and $(v,v_k)$ are home-edges which can be counted by $w(I^H_{v,1}\cap \delta_H(v))$, but the edges $(v,v_i)$ with $2\leq i\leq k-1$ are not home-edges.
By the triangle inequality, we know that $w(v,v_i)\leq \frac{1}{2}w(C)$. Therefore, the weight of the uncounted edges are at most $\frac{k-2}{2}w(C)$. Thus, we have
$$\sum_{u\in S_{v,1}} w(u,v)\leq \lrA{\frac{k-2}{2}+\alpha_v}\gamma_v w(I^H_v).$$

Similarly, $\beta_v(1-\gamma_v)w(I^H_v) = w(I^H_{v,2}\cap \delta_H(v))$ is the weight of all home-edges in $k'$-cycles. For any $k'$-cycle $C=(v,v_1,\dots,v_{k'-1},v)$ in $I^H_{v,2}$, the home-edges $(v,v_1)$ and $(v,v_{k'-1})$ can be counted by $w(I^H_{v,2}\cap \delta_H(v))$.
The weight of the uncounted edges are at most $\frac{k'-3}{2}w(C)\leq \frac{k-3}{2}w(C)$.
Thus, we have
$$\sum_{u\in S_{v,2}} w(u,v)\leq \lrA{\frac{k-3}{2}+\beta_v}\gamma_v w(I^H_v).$$

Then, we have
\begin{align*}
&\lrA{\frac{k-2}{2}+\alpha_v}\gamma_v w(I^H_v)+\lrA{\frac{k-3}{2}+\beta_v}(1-\gamma_v) w(I^H_v)\\
&\geq \sum_{u\in S_{v,1}} w(u,v)+\sum_{u\in S_{v,2}}w(u,v)=\deg_H(v).
\end{align*}
\end{proof}

\begin{lemma}\label{lb-tree-}
It holds that $(1-\frac{1}{2}\alpha_v)\gamma_v w(I^H_v)+(1-\frac{1}{2}\beta_v)(1-\gamma_v) w(I^H_v) \geq \mbox{MST}(H)$.
\end{lemma}
\begin{proof}
In the optimal itinerary of $v$ on $H$, the total weight of all home-edges is $$\alpha_v\gamma_v w(I^H_v)+\beta_v(1-\gamma_v)w(I^H_v).$$
For each cycle, there are exactly two home-edges. We can delete the longer edge from each cycle from $I^H_v$,
and get a spanning tree with a total weight of at most
$$\lrA{1-\frac{1}{2}\alpha_v}\gamma_v w(I^H_v)+\lrA{1-\frac{1}{2}\beta_v}(1-\gamma_v) w(I^H_v).$$
Since the weight of a minimum spanning tree in graph $H$ is $\mbox{MST}(H)$, we have that
$$\lrA{1-\frac{1}{2}\alpha_v}\gamma_v w(I^H_v)+\lrA{1-\frac{1}{2}\beta_v}(1-\gamma_v)w(I^H_v)\geq \mbox{MST}(H).$$
The lemma holds.
\end{proof}

The new bounds in Lemmas \ref{lb-delta-} and \ref{lb-tree-} are stronger than that in \cite{yamaguchi2009improved}. Next, we propose two new lower bounds on the minimum weight perfect matching and $k$-cycle/path packing. The lower bound related to the minimum weight perfect matching works only for TTP-3.
\begin{lemma}\label{lb-matching-}
Let $M_H$ be a minimum weight perfect matching in graph $H$. For TTP-3, we have that $\frac{1}{2}\gamma_v w(I^H_v)+(1-\beta_v)(1-\gamma_v) w(I^H_v)\geq w(M_H)$.
\end{lemma}
\begin{proof}
We assume that $k=3$ and consider TTP-3. Note that $I^H_{v,1}$ is a set of 4-cycles and $I^H_{v,2}$ is a set of 3-cycles.
Since $n_S=2mk=6m$, we know that the number of vertices in $S_{v,1}$ is odd but even in $S_{v,2}$. Since any pair of 3-cycles only share one common vertex $v$, after we delete both of home-edges of each 3-cycle, we can get a perfect matching $M'_H$ in graph $H[S_{v,2}]$ with a total weight of $(1-\beta_v)(1-\gamma_v)w(I^H_v)$.

Then, by a similar argument in the proof of Lemma \ref{lb-tsp}, we know that $\gamma_v w(I^H_v)=w(I^H_{v,1})$ is greater than the weight of the minimum weight Hamiltonian cycle in graph $H[S_{v,1}\cup\{v\}]$. Since the number of vertices in this graph is even, any Hamiltonian cycle in this graph can be decomposed into two perfect matchings. Thus, we can get a perfect matching $M''_H$ in this graph with a total weight of at most $\frac{1}{2}\gamma_v w(I^H_v)$. Therefore, $M'_H\cup M''_H$ is a perfect matching in graph $H$, and then we have $w(M_H)\leq w(M'_H\cup M''_H)\leq \frac{1}{2}\gamma_v w(I^H_v)+(1-\beta_v)(1-\gamma_v)w(I^H_v)$.
\end{proof}

\begin{lemma}\label{lb-packing-}
For graph $H$, let $\P^*_k$ be a minimum weight $k$-path packing, and $\C^*_k$ be a minimum weight $k$-cycle packing. It holds that $\lrA{\frac{k-1}{k}+\frac{1}{k}\gamma_v-\alpha_v\gamma_v}w(I^H_v)=(1-\alpha_v)\gamma_v w(I^H_v)+\frac{k-1}{k}(1-\gamma_v)w(I^H_v)\geq w(\P^*_k)\geq \frac{1}{2}w(\C^*_k)$.
\end{lemma}
\begin{proof}
First, we consider $w(\P^*_k)\geq \frac{1}{2}w(\C^*_k)$.
Let $\C'_k$ be the $k$-cycle packing obtained by completing the $k$-path packing.
For any $k$-path in $\P^*_k$, saying $P=(v_1,\dots,v_k)$, we obtain $w(P)\geq w(v_1,v_k)$ by the triangle inequality. This shows $2w(\P^*_k)\geq w(\C'_k)\geq w(\C^*_k)$, and we are done.

We note that the number of vertices is divisible by $k$ in $S_{v,1}$ but not in $S_{v,2}$. Since any pair of $(k+1)$-cycles only share one common vertex $v$, after we delete both of home-edges for each $(k+1)$-cycle, we can get a $k$-path packing $\P'_k$ in graph $H[S_{v,1}]$ such that $w(\P'_k) = (1-\alpha_v)\gamma_v w(I^H_v)$.

Then, by a similar argument in the proof of Lemma \ref{lb-tsp}, we know that $(1-\gamma_v)w(I^H_v) = w(I^H_{v,2}) \geq w(C)$ where $C$ is the minimum weight Hamiltonian cycle in graph $H[S_{v,2}\cup\{v\}]$. Since the number of vertices in this graph equals $n_S$ minus the number of vertices in $S_{v,1}$, which is divisible by $k$, it is easy to see that we can delete an edge every $k-1$ edges along $C$ to get a $k$-path packing $\P''_k$ such that $w(\P''_k)\leq \frac{k-1}{k} w(C)$.
Note that $\P'_k\cup \P''_k$ is a $k$-path packing in graph $H$. Thus, we have that
$(1-\alpha_v)\gamma_v w(I^H_v)+\frac{2}{3}(1-\gamma_v) w(I^H_v) \geq w(\P'_k\cup \P''_k) \geq w(\P^*_k)$.
\end{proof}

Then, we propose a new lower bound related to the minimum weight mod $k$-cycle packing, a stronger lower bound related to the minimum weight $k$-cycle packing, and a lower bound related to the minimum weight perfect matching for any $k\geq 3$. 

\begin{lemma}\label{lb-pack1-}
For graph $H$, let $\C^*_{\bmod k}$ be a minimum weight mod $k$-cycle packing. It holds that $\min\{2(1-\alpha_v),1\}\cdot\gamma_v w(I^H_v)+(1-\gamma_v)w(I^H_v)\geq w(\C^*_{\bmod k})$.
\end{lemma}
\begin{proof}
We first construct a $k$-cycle packing in $H[S_{v,1}]$.
For any $(k+1)$-cycle $C=(v,v_1,\dots,v_k,v)$ in $I^H_{v,1}$, we can get a $(k+1)$-cycle $C'=(v_1,\dots,v_k,v_1)$ with a weight of $\sum_{i=1}^{k-1}w(v_i, v_{i+1})+w(v_1, v_k)=w(C)-w(v,v_1)-w(v,v_k)+w(v_1,v_k)$. By the triangle inequality, we have
\[
w(v_1,v_k)\leq \min\{w(C)-w(v,v_1)-w(v,v_k),\  w(v,v_1)+w(v,v_k)\}.
\]
Hence, we have
\[
w(C')\leq \min\{2(w(C)-w(v,v_1)-w(v,v_k)),\ w(C)\}.
\]
Since $(v,v_1)$ and $(v,v_k)$ are home-edges, we can get a $k$-cycle packing $\C'_k$ in $H[S_{v,1}]$ such that
\begin{align*}
w(\C'_k)&\leq \min\{2(w(I^H_{v,1})-w(I^H_{v,1}\cap \delta_H(v))),\ w(I^H_{v,1})\}\\
&=\min\{2(1-\alpha_v)w(I^H_{v,1}),\ w(I^H_{v,1})\}\\
&=\min\{2(1-\alpha_v),1\}\cdot \gamma_vw(I^H_v),
\end{align*}
where the first equality follows from $w(I^H_{v,1}\cap \delta_H(v))=\alpha_vw(I^H_{v,1})$, and the second follows from $w(I^H_{v,1})=\gamma_vw(I^H_{v})$.

Then, by the proof of Lemma~\ref{lb-packing-}, we know that $(1-\gamma_v)w(I^H_v) \geq w(C'')$, where $C''$ is the minimum weight Hamiltonian cycle in graph $H[S_{v,2}\cup\{v\}]$ and the number of vertices in $C$ is divisible by $k$. Hence, we can get a mod $k$-cycle packing $\C''_{\bmod k}=\{C''\}$ in $H[S_{v,2}\cup\{v\}]$ such that
\[
w(\C''_{\bmod k}) \leq (1-\gamma_v)w(I^H_v).
\]

Since $\C'_k\cup\C''_{\bmod k}$ is a mod $k$-cycle packing in $H$ and $\C^*_{\bmod k}$ is a minimum weight mod $k$-cycle packing in $H$, $\C^*_{\bmod k}$ has a weight of at most $\min\{2(1-\alpha_v),1\}\cdot\gamma_v w(I^H_v)+(1-\gamma_v)w(I^H_v)$.
\end{proof}

\begin{lemma}\label{lb-pack2-}
For graph $H$, let $\C^*_k$ be a minimum weight $k$-cycle packing. It holds that $\min\{2(1-\alpha_v),1\}\cdot\gamma_v w(I^H_v)+\frac{2(k-1)}{k}(1-\gamma_v)w(I^H_v)\geq w(\C^*_k)$.
\end{lemma}
\begin{proof}
By the proof of Lemma~\ref{lb-pack1-}, we can get a $k$-cycle packing $\C'_k$ in $H[S_{v,1}]$ with a weight of at most $\min\{2(1-\alpha_v),1\}\cdot \gamma_vw(I^H_v)$.
By the proof of Lemma~\ref{lb-packing-}, we can get a $k$-cycle packing $\C''_k$ with a weight of at most $\frac{2(k-1)}{k}(1-\gamma_v)w(I^H_v)$.
Therefore, there is a $k$-cycle packing $\C'_k\cup\C''_k$ in $H$ with a weight of at most $\min\{2(1-\alpha_v),1\}\cdot\gamma_v w(I^H_v)+\frac{2(k-1)}{k}(1-\gamma_v)w(I^H_v)$. Since $\C^*_k$ is the minimum weight $k$-cycle packing in $H$, the lemma holds.
\end{proof}

\begin{lemma}\label{lb-matching1-}
Let $C_H$ be a minimum weight Hamiltonian cycle, $\C^*_{\bmod k}$ is a minimum weight mod $k$-cycle packing, and $M_H$ be a minimum weight perfect matching in graph $H$. For TTP-$k$ with $k\geq 3$, we have $w(H)\geq 2w(M_H)$. Moreover, if $k$ is even, we have $w(\C^*_{\bmod k})\geq 2w(M_H)$.
\end{lemma}
\begin{proof}
Since the number of vertices in $H$ is even, the edges of $C_H$ can be decomposed into two perfect matchings. The one with a smaller weight has a weight of at most $w(C_H)$. So, $w(H)\geq 2w(M_H)$.

If $k$ is even, the number of vertices on each cycle $C\in \C^*_{\bmod k}$ is also even. Similarly, for each cycle $C\in \C^*_{\bmod k}$, we can obtain a perfect matching $M_C$ in the graph induced by the vertices of $C$ with a weight of at most $\frac{1}{2}w(C)$. Hence, we obtain a perfect matching in $H$ with a weight of at most $\frac{1}{2}w(\C^*_{\bmod k})$. So, we have $w(\C^*_{\bmod k})\geq 2w(M_H)$.
\end{proof}

\subsection{Independent Lower Bounds}
Now, we are ready to consider the independent lower bounds, which are found by summing the individual optimal itineraries of teams in $S$ on graph $H$.

For each team $v\in S$, $I_v$ (resp., $I^H_v$) is the optimal itinerary of team $v$ on graph $G$ (resp., $H$).
Recall that $\psi=\sum_{v\in V}\psi_v$ and $\psi_v=w(I_v)$. We define
$$\psi^H=\sum_{v\in S}\psi^H_v \mbox{~~and~~} \psi^H_v=w(I^H_v).$$ Then, we can get
$$\psi=\sum_{v\in V}w(I_v)\geq \sum_{v\in S}w(I_v)\geq \sum_{v\in S}w(I^H_v)=\psi^H,$$ i.e., $\psi^H$ is a lower bound of TTP-$k$.

Define $\gamma$, $\alpha$ and $\beta$ such that it holds
$$\gamma\psi^H = \sum_{v\in S} w(I^H_{v,1})=\sum_{v\in S}\gamma_vw(I^H_v),$$
$$\alpha\gamma\psi^H = \sum_{v\in S} w(\delta_H(v)\cap I^H_{v,1})=\sum_{v\in S}\alpha_v\gamma_vw(I^H_v),$$ and
$$\beta(1-\gamma)\psi^H = \sum_{v\in S} w(\delta_H(v)\cap I^H_{v,2})=\sum_{v\in S}\beta_v(1-\gamma_v)w(I^H_v).$$
We have $0\leq \gamma, \alpha, \beta\leq 1$.

The lemmas we proved previously can be summed together to obtain different lower bounds.

\begin{lemma}\label{lb-tsp+}
Let $C_H$ be a minimum weight Hamiltonian cycle in graph $H$. It holds that $$\psi\geq n_Sw(C_H).$$
\end{lemma}
\begin{proof}
Since $\psi\geq\psi^H=\sum_{v\in S}w(I^H_v)$, by Lemma \ref{lb-tsp}, we have $\psi\geq\sum_{v\in S}w(I^H_v)\geq n_Sw(C_H)$.
\end{proof}

\begin{lemma}\label{lb-delta}
$(\frac{k-1}{2}-\frac{1}{2}\gamma+\alpha\gamma)\psi\geq(\frac{k-2}{2}+\alpha)\gamma\psi +(\frac{k-3}{2}+\beta)(1-\gamma)\psi\geq \Delta_H$.
\end{lemma}
\begin{proof}
Since $\psi\geq\psi^H$, it is sufficient to prove
$$\lrA{\frac{k-1}{2}-\frac{1}{2}\gamma+\alpha\gamma}\psi^H\geq\lrA{\frac{k-2}{2}+\alpha}\gamma\psi^H +\lrA{\frac{k-3}{2}+\beta}(1-\gamma)\psi^H\geq \Delta_H.$$

By the definitions of $\gamma$, $\alpha$ and $\beta$, we have
\[
\lrA{\frac{k-1}{2}-\frac{1}{2}\gamma+\alpha\gamma}\psi^H=\sum_{v\in S}\lrA{\frac{k-1}{2}-\frac{1}{2}\gamma_v+\alpha_v\gamma_v}w(I^H_v)
\]
and
\begin{align*}
&\lrA{\frac{k-2}{2}+\alpha}\gamma\psi^H +\lrA{\frac{k-3}{2}+\beta}(1-\gamma)\psi^H\\
&\ =\sum_{v\in S}\lrA{\lrA{\frac{k-2}{2}+\alpha_v}\gamma_vw(I^H_v) +\lrA{\frac{k-3}{2}+\beta_v}(1-\gamma_v)w(I^H_v)}.
\end{align*}
Recall that $\Delta_H=\sum_{v\in S}\deg_H(v)$.
By Lemma \ref{lb-delta-}, it holds that
\begin{align*}
&\sum_{v\in S}\lrA{\frac{k-1}{2}-\frac{1}{2}\gamma_v+\alpha_v\gamma_v}w(I^H_v)\\
&\geq\sum_{v\in S}\lrA{\lrA{\frac{k-2}{2}+\alpha_v}\gamma_v w(I^H_v)+\lrA{\frac{k-3}{2}+\beta_v}(1-\gamma_v) w(I^H_v)}\\
&\geq\sum_{v\in S}\deg_H(v)=\Delta_H.
\end{align*}
Then, we have $$\lrA{\frac{k-1}{2}-\frac{1}{2}\gamma+\alpha\gamma}\psi^H\geq\lrA{\frac{k-2}{2}+\alpha}\gamma\psi^H +\lrA{\frac{k-3}{2}+\beta}(1-\gamma)\psi^H\geq \Delta_H.$$
\end{proof}
Note that $\Delta_H\leq (\frac{k-1}{2}-\frac{1}{2}\gamma+\alpha\gamma)\psi\leq \frac{k}{2}\psi$. The lower bound $\Delta_H\leq \frac{k}{2}\psi$ can be extended to the graph $G$. We have the following lemma.
\begin{lemma}[\cite{westphal2014}]\label{lb-delta+}
$\frac{k}{2}\psi\geq \Delta_G$.
\end{lemma}

Note that $\Delta_G\leq \Delta_H+2\sum_{v\in \overline{S}}\deg_G(v)$. By (\ref{degree+}) and Lemma \ref{lb-delta+}, we have that
\begin{equation}\label{deltagh}
\Delta_G\leq \Delta_H+\frac{4k^2}{n}\psi.
\end{equation}

\begin{lemma}\label{lb-tree}
It holds that
\[
\left(1-\frac{1}{2}\alpha\right)\gamma\psi+\left(1-\frac{1}{2}\beta\right)(1-\gamma)\psi \geq n_S\mbox{MST}(H).
\]
\end{lemma}
\begin{proof}
We will prove $(1-\frac{1}{2}\alpha)\gamma\psi^H+(1-\frac{1}{2}\beta)(1-\gamma)\psi^H \geq n_S\mbox{MST}(H)$.

Note that
\begin{align*}
&\lrA{1-\frac{1}{2}\alpha}\gamma\psi^H+\lrA{1-\frac{1}{2}\beta}(1-\gamma)\psi^H\\
&=\sum_{v\in S}\lrA{\lrA{1-\frac{1}{2}\alpha_v}\gamma_v w(I^H_v)+\lrA{1-\frac{1}{2}\beta_v}(1-\gamma_v)w(I^H_v)}.
\end{align*}
By Lemma \ref{lb-tree-}, we know that
\[
\sum_{v\in S}\lrA{\lrA{1-\frac{1}{2}\alpha_v}\gamma_v w(I^H_v)+\lrA{1-\frac{1}{2}\beta_v}(1-\gamma_v)w(I^H_v)}\geq n_S\mbox{MST}(H).
\]
Thus, we have
\[
\lrA{1-\frac{1}{2}\alpha}\gamma\psi^H+\lrA{1-\frac{1}{2}\beta}(1-\gamma)\psi^H\geq n_S\mbox{MST}(H).
\]
\end{proof}

By Lemmas \ref{lb-delta} and \ref{lb-tree}, we have that
\begin{equation}\label{lb-delta_tau}
\lrA{\frac{k+1}{4}+\frac{1}{4}\gamma}\psi\geq \frac{1}{2}\Delta_H+n_S\mbox{MST}(H).
\end{equation}

\begin{lemma}\label{lb-matching}
Let $M_H$ be a minimum weight perfect matching in graph $H$. For TTP-3, we have that
$$\frac{1}{2}\gamma\psi+(1-\beta)(1-\gamma)\psi\geq n_Sw(M_H).$$
\end{lemma}
\begin{proof}
We will prove $\frac{1}{2}\gamma\psi^H+(1-\beta)(1-\gamma)\psi^H+\geq n_Sw(M_H)$.

Note that $$\frac{1}{2}\gamma\psi^H+(1-\beta)(1-\gamma)\psi^H=\sum_{v\in S}\lrA{\frac{1}{2}\gamma_v w(I^H_v)+(1-\beta_v)(1-\gamma_v)w(I^H_v)}.$$
By Lemma~\ref{lb-matching-}, we know that
$$\sum_{v\in S}\lrA{\frac{1}{2}\gamma_v w(I^H_v)+(1-\beta_v)(1-\gamma_v)w(I^H_v)}\geq n_Sw(M_H).$$ Thus, we have $$\frac{1}{2}\gamma\psi^H+(1-\beta)(1-\gamma)\psi^H\geq n_Sw(M_H).$$
\end{proof}

By Lemmas \ref{lb-delta} and \ref{lb-matching}, for TTP-3, we have that
\begin{equation}\label{lb-delta_matching}
(1+\alpha\gamma)\psi\geq \Delta_H+n_Sw(M_H)
\end{equation}
for a minimum weight perfect matching $M_H$ in graph $H$.

\begin{lemma}\label{lb-packing}
For graph $G$, let $\P^*_k$ be a minimum weight $k$-path packing, and $\C^*_k$ be a minimum weight $k$-cycle packing. It holds that $$\lrA{\frac{k-1}{k}+\frac{1}{k}\gamma-\alpha\gamma}\psi=(1-\alpha)\gamma\psi+\frac{k-1}{k}(1-\gamma)\psi\geq n_Sw(\P^*_k)\geq \frac{1}{2}n_Sw(\C^*_k).$$
\end{lemma}
\begin{proof}
We will prove $(\frac{k-1}{k}+\frac{1}{k}\gamma-\alpha\gamma)\psi^H=(1-\alpha)\gamma\psi^H+\frac{k-1}{k}(1-\gamma)\psi^H\geq n_Sw(\P^*_k)\geq \frac{1}{2}n_Sw(\C^*_k)$.

Note that $$\lrA{\frac{k-1}{k}+\frac{1}{k}\gamma-\alpha\gamma}\psi^H=\sum_{v\in S}\lrA{\frac{k-1}{k}+\frac{1}{k}\gamma_v-\alpha_v\gamma_v}w(I^H_v)$$ and
$$(1-\alpha)\gamma\psi^H+\frac{k-1}{k}(1-\gamma)\psi^H=\sum_{v\in S}\lrA{(1-\alpha_v)\gamma_vw(I^H_v)+\frac{k-1}{k}(1-\gamma_v)w(I^H_v)}.$$
By Lemma \ref{lb-packing-}, it holds that
\begin{align*}
\sum_{v\in S}\lrA{\frac{k-1}{k}+\frac{1}{k}\gamma_v-\alpha_v\gamma_v}w(I^H_v)\geq&\ \sum_{v\in S}\lrA{(1-\alpha_v)\gamma_vw(I^H_v)+\frac{k-1}{k}(1-\gamma_v)w(I^H_v)}\\
\geq&\ n_Sw(\P^*_k)\geq \frac{1}{2}n_Sw(\C^*_k).
\end{align*}
Therefore, we have that $$\lrA{\frac{k-1}{k}+\frac{1}{k}\gamma-\alpha\gamma}\psi^H=(1-\alpha)\gamma\psi^H+\frac{k-1}{k}(1-\gamma)\psi^H\geq n_Sw(\P^*_k)\geq \frac{1}{2}n_Sw(\C^*_k).$$
\end{proof}

\begin{lemma}\label{lb-pack1}
For graph $H$, let $\C^*_{\bmod k}$ be a minimum weight mod $k$-cycle packing. It holds that $$\min\{2(1-\alpha),1\}\cdot\gamma \psi+(1-\gamma)\psi\geq n_Sw(\C^*_{\bmod k}).$$
\end{lemma}
\begin{proof}
The proof follows from a similar analysis in the proof of Lemma~\ref{lb-packing}.
\end{proof}

\begin{lemma}\label{lb-pack2}
For graph $H$, let $\C^*_k$ be a minimum weight $k$-cycle packing. It holds that
\[
\min\{2(1-\alpha),1\}\cdot\gamma \psi+\frac{2(k-1)}{k}(1-\gamma)\psi\geq n_Sw(\C^*_k).
\]
\end{lemma}
\begin{proof}
The proof follows from a similar analysis in the proof of Lemma~\ref{lb-packing}.
\end{proof}

\begin{lemma}\label{lb-matching1}
Let $M_H$ be a minimum weight perfect matching in graph $H$. For TTP-$k$ with $k\geq 3$, we have $\frac{1}{2}\psi\geq n_Sw(M_H)$. Moreover, if $k$ is even, we have
\[
\frac{1}{2}\cdot\min\{2(1-\alpha),1\}\cdot\gamma \psi+\frac{1}{2}(1-\gamma)\psi\geq n_Sw(M_H).
\]
\end{lemma}
\begin{proof}
It follows directly from Lemmas~\ref{lb-matching1-}, \ref{lb-tsp+}, and \ref{lb-pack1}.
\end{proof}

Next, we will describe our algorithm. Our algorithm consists of two constructions, where the first is based on the Hamiltonian cycle and the second is based on the $k$-cycle/path packing.
The quality of each will be analyzed after showing the construction.
At last, we will analyze the approximation quality of each, make a trade-off between them, and then get the final approximation ratio.

\section{The Constructions}\label{B}
In this section, we introduce the methods to construct feasible schedules for TTP-$k$. 
We have two construction methods.
The first construction is based on a given Hamiltonian cycle of the graph. This method has been widely used in the literature and we will not repeat all the details.
The second construction is based on a given cycle packing of the graph. This can be considered as a new method.
We need to use the two methods to do some trade-offs in the analysis.

\subsection{The Hamiltonian Cycle Construction}\label{B.1}
The idea of Hamiltonian cycle construction is to make use of the canonical schedule~\cite{kirkman1847problem,de1981scheduling} and a Hamiltonian cycle.
For TTP-$k$, many previous schedules were constructed by using this method~\cite{yamaguchi2009improved,westphal2014,imahori2010approximation}. In our algorithm, we will also directly use the well-analyzed schedule in~\cite{yamaguchi2009improved}. 
Next, we give a brief introduction to this construction.

Roughly speaking, this schedule is generated by a rotation scheme that can make sure that almost all road trips of each team are $(k+1)$-cycles and in each road trip, the team visits a set of consecutive teams along the Hamiltonian cycle used.
%
%
For the quality of this construction, we have the following lemma.
\begin{lemma}[\cite{yamaguchi2009improved}]\label{lb-distance}
Let $C$ be a Hamiltonian cycle in graph $G$. For TTP-$k$ with any constant $k\geq3$, there is a polynomial-time algorithm that can generate a solution with a total weight of at most
$$\frac{k-1}{k}n w(C)+\frac{2}{k}\Delta_G+O(k/n)\psi.$$
\end{lemma}

We cannot use this result directly since most of our previous lower bounds are based on the graph $H$.
To get a similar result based on a given Hamiltonian cycle in graph $H$, we prove the following lemma.
\begin{lemma}\label{lb-distance+}
Let $C_H$ be a Hamiltonian cycle in graph $H$ such that $n_Sw(C_H)\leq O(1)\psi$. For TTP-$k$ with any constant $k\geq3$, there is a polynomial-time algorithm that can generate a solution with a total weight of at most
\[
\frac{k-1}{k}n_S w(C_H)+\frac{2}{k}\Delta_H+O(k^2/n)\psi.
\]
\end{lemma}
\begin{proof}
Given a Hamiltonian cycle $C_H$ in graph $H$, we can construct a graph by adding two copies of edge $(u,v)$ for each $u\in \overline{S}$, where $v=\arg\min_{v'\in S} w(u,v')$. So, it holds that $$w(u,v)\leq \frac{1}{n_S}\sum_{v'\in S}w(u,v').$$ The new graph is a connected Eulerian graph on vertex set $V$, and the total weight of edges in this graph is at most
\begin{equation}\label{HG}
w(C_H)+\sum_{u\in \overline{S}}\lrA{\frac{2}{n_S}\sum_{v\in S}w(u,v)}\leq
w(C_H)+\frac{2}{n_S}\sum_{u\in \overline{S}}\deg_G(u)\leq
w(C_H)+\frac{16k}{n^2}\Delta_G,
\end{equation}
where the first inequality follows from $\sum_{v\in S}w(u,v)\leq \sum_{v\in V}w(u,v)=\deg_G(u)$, and the second follows from Lemma~\ref{degree+} and (\ref{n_S}). 
By shortcutting this Eulerian graph, we can get a Hamiltonian cycle $C$ in graph $G$ with a weight of at most
\[
w(C_H)+\frac{16k}{n^2}\Delta_G\leq w(C_H)+\frac{8k^2}{n^2}\psi
\]
by Lemma~\ref{lb-delta+} and (\ref{HG}). Using this cycle, by Lemma~\ref{lb-distance}, we can generate a solution with a total weight of at most
\begin{align*}
&\frac{k-1}{k}n\lrA{w(C_H)+\frac{8k^2}{n^2}\psi}+\frac{2}{k}\Delta_G+O(k/n)\psi\\
&\leq \frac{k-1}{k}nw(C_H)+\frac{2}{k}\Delta_G+\frac{8k^2}{n}\psi+O(k/n)\psi.
\end{align*}

Note that
\[
\frac{k-1}{k}nw(C_H)\leq \frac{k-1}{k}n_Sw(C_H)+n_{\overline{S}}w(C_H)\leq\frac{k-1}{k}n_Sw(C_H)+\frac{8k}{n}n_Sw(C_H),
\]
where the first inequality follows from $n=n_S+n_{\overline{S}}$, and the second follows from $n_{\overline{S}}<4k$ and $1/n_S<2/n$ by (\ref{n_S}).
By (\ref{deltagh}), we also have
\[
\frac{2}{k}\Delta_G=\frac{2}{k}\Delta_H+\frac{8k}{n}\psi.
\]
Based on the condition that $n_Sw(C_H)\leq O(1)\psi$, we know that the total weight of the solution is at most
\begin{align*}
&\frac{k-1}{k}nw(C_H)+\frac{2}{k}\Delta_G+\frac{8k^2}{n}\psi+O(k/n)\psi\\
&\ \leq \frac{k-1}{k}n_Sw(C_H)+\frac{2}{k}\Delta_H+\frac{8k}{n}n_Sw(C_H)+\frac{8k}{n}\psi+\frac{8k^2}{n}\psi+O(k/n)\psi\\
&\ \leq \frac{k-1}{k}n_Sw(C_H)+\frac{2}{k}\Delta_H+O(k^2/n)\psi.
\end{align*}
\end{proof}

Note that the condition $n_Sw(C_H)\leq O(1)\psi$ in Lemma~\ref{lb-distance+} is easy to satisfy.
For example, if the Hamiltonian cycle $C_H$ is obtained by the Christofides-Serdyukov algorithm (a $3/2$-approximation algorithm for TSP)~\cite{christofides1976worst,serdyukov1978some}, then we have $n_Sw(C_H)\leq \frac{3}{2}\psi=O(1)\psi$ by Lemma~\ref{lb-tsp+}.

\subsection{The Cycle Packing Construction}\label{B.2}
In this subsection, we will construct a feasible schedule based on a given $k$-cycle packing of the graph.
We will also show that the construction works for a given $k$-path packing.
One challenge for the  construction is to build a feasible schedule for $n$ not being divisible by $k$. We introduce the construction in detail.

\subsubsection{The construction}
Given a $k$-cycle packing $\C_k=\{C_1,\dots,C_{2m}\}$ in graph $H$. Recall that the $k$-cycles were arbitrarily labeled. There are $2m$ cycle-teams $\{u_1,\dots, u_{2m}\}$ and $m$ super-teams $\{U_1,\dots,U_m\}$, where $U_i=u_{2i-1}\cup u_{2i}$ and $u_i=\{t_{k(i-1)+1}, \dots, t_{k(i-1)+k}\}$. There are still $n_{\overline{S}}$ normal teams in $\overline{S}$, which will be arbitrarily divided into $n_{\overline{S}}/2$ pairs, denoted by $r$ \emph{team-pairs} $\{R_1,\dots,R_r\}$, where $R_{i}=\{t_{n_S+2i-1}, t_{n_S+2i}\}$ and $r=n_{\overline{S}}/2$.
The packing construction is to first arrange \emph{super-games} for super-teams (including the $r$ team-pairs) and then extend the super-games into a double round-robin for normal teams.

Recall that $n\geq 8k^2$, $n_{S}=2mk$, and $n_{\overline{S}}=n-n_{S}$.
Hence, we have $m=2\floor{\frac{n}{4k}}\geq4k>n_{\overline{S}}=2r$.
Each super-team will first attend $m-1$ super-games in $m-1$ time slots.
Each super-game will be extended to normal games between normal teams that span $4k$ days.
So each normal team will attend $4k(m-1)$ normal games.
Note that each normal team should attend $2n-2$ normal games in TTP-$k$.
Recall that $n=2mk+2r$.
Hence, there are still $2n-2-4k(m-1)=4k+4r-2$ days of normal games, called \emph{unarranged normal games}, which will be designed after the super-games in the last time slot.

We will construct the schedule for super-teams from the first time slot to the last time slot.
In each of the $m-1$ time slots, we have $m/2$ super-games.
The schedule in the last time slot, which also contains the design of unarranged normal games, is different from the schedules in the first $m-2$ time slots.

For the first time slot, the $m/2$ super-games are arranged as shown in Figure~\ref{figa}. The most left super-game including two super-teams is called \emph{left super-game}, it contains two super-teams, and we put a letter `L' on the edge to indicate it.
The most right $r$ super-games are called \emph{right super-games}, each of them contains two super-teams and one team-pair, and we put a letter `R' on the edge to indicate it.
The middle $m/2-r-1$ super-games are called \emph{normal super-games} and each of them contains two super-teams.
Each super-game contains a directed edge between the two super-teams, where a directed edge from $U_i$ to $U_j$ means a super-game between them at the home of $U_j$.

\begin{figure}[ht]
    \centering
    \includegraphics[scale=0.7]{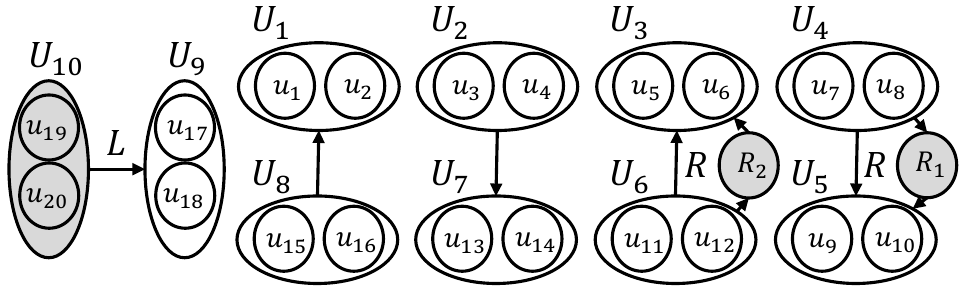}
    \caption{The schedule in the first time slot for an instance with $m=10$ and $r=2$}
    \label{figa}
\end{figure}

In Figure~\ref{figa}, we can see that the last super-team $U_m$ and the $r$ team-pairs are denoted as dark nodes and all other super-teams
$\{U_1, \dots, U_{m-1}\}$ are denoted as white nodes which form a cycle $(U_1, \dots, U_{m-1}, U_1)$.
In the second time slot, we fix the positions of dark nodes and change the positions of white super-teams in the cycle by moving one position in the clockwise direction, and we also change the direction of each edge.
In the second time slot, there is still $1$ left super-game, $m/2-r-1$ normal super-games and $r$ right super-games.
An illustration of the schedule in the second time slot is shown in Figure~\ref{figb}.

 \begin{figure}[ht]
    \centering
    \includegraphics[scale=0.7]{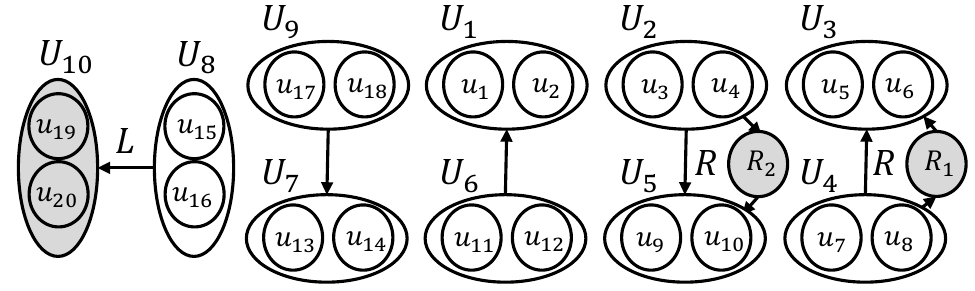}
    \caption{The schedule in the second time slot for an instance with $m=10$ and $r=2$}
    \label{figb}
\end{figure}

The schedules for the other middle slots are derived analogously.
Before we introduce the super-games in the last time slot, we first explain how to extend the super-games in the first $m-2$ time slots to normal games.
In these time slots, we have three different kinds of super-games: normal super-games, left super-games and right super-games. We first consider normal super-games.

\textbf{Normal super-games}:
Assume that there is a normal super-game $U_i\rightarrow U_j$, i.e., a normal super-game between super-teams $U_i$ and $U_j$ at the home of $U_j$. Recall that $U_{i}$ represents cycle-teams $u_{2i-1}\cup u_{2i}$ and $U_{j}$ represents cycle-teams $u_{2j-1}\cup u_{2j}$. The normal super-game will be extended to \emph{cycle-games} in two \emph{sub-time slots}: in the first sub-time slot, there are two \emph{normal cycle-games}: $\{u_{2i-1}\rightarrow u_{2j-1},\  u_{2i}\rightarrow u_{2j}\}$, where $u_{i'}\rightarrow u_{i''}$ means a normal cycle-game between cycle-teams $u_{i'}$ and $u_{i''}$ at the home of $u_{i''}$; in the second sub-time slot, there are still two normal cycle-games: $\{u_{2i-1}\rightarrow u_{2j},\  u_{2i}\rightarrow u_{2j-1}\}$.

Next, we only need to explain the design of normal cycle-games.

Assume that there is a normal cycle-game $u_i\rightarrow u_j$, i.e., a normal cycle-game between cycle-teams $u_i$ and $u_j$ at the home of $u_j$. Recall that $u_{i}$ represents normal teams $\{t_{k(i-1)+1}, \dots, t_{k(i-1)+k}\}$ and $u_{j}$ represents normal teams $\{t_{k(j-1)+1}, \dots, t_{k(j-1)+k}\}$. For the sake of presentation, we let $t_{k(i-1)+i'}=x_{i'}$ and $t_{k(j-1)+i'}=y_{i'}$, i.e., $u_i=\{x_1,\dots, x_k\}$ and $u_j=\{y_1,\dots, y_k\}$.
Let $p_l$ denote the $k$ normal games in $\{x_{i}\rightarrow y_{(i+l-2)\bmod k)+1}\}_{i=1}^{k}$, where $x_{i'}\rightarrow y_{i''}$ means a normal game between normal teams $x_{i'}$ and $y_{i''}$ at the home of $y_{i''}$. In the following, we also use $\overline{p_l}$ to denote the $k$ normal games of $p_l$, but with all venues reversed.
The normal cycle-game will be extended to normal games on $2k$ days, which can
can be presented by
\[
\lrA{p_1\cdots p_{k-2}\cdot p_{k-1}\cdot p_k}\cdot\lrA{\overline{p_1\cdots  p_{k-2}\cdot p_{k-1} \cdot p_k}}.
\]

An example of a normal super-game with $k=3$ is shown in Table~\ref{expand_normal}.
We can see that each normal super-game will be extended to normal games on $4k$ days. There are at most $k$-consecutive home/away normal games. The normal games between normal teams in $U_i$ and normal teams in $U_j$, which can be seen as a directed complete bipartite graph $B(U_i,U_j)$, are all arranged.

\begin{table}[ht]
\centering
\begin{tabular}{c|cccccccccccc}
  & $1$ & $2$ & $3$ & $4$ & $5$ & $6$ & $7$ & $8$ & $9$ & $10$ & $11$ & $12$\\
\hline
  $x_1$ & $y_1$ & $y_2$ & $y_3$ & $\pmb{y_1}$ & $\pmb{y_2}$ & $\pmb{y_3}$
        & $y_4$ & $y_5$ & $y_6$ & $\pmb{y_4}$ & $\pmb{y_5}$ & $\pmb{y_6}$\\
  $x_2$ & $y_2$ & $y_3$ & $y_1$ & $\pmb{y_2}$ & $\pmb{y_3}$ & $\pmb{y_1}$
        & $y_5$ & $y_6$ & $y_4$ & $\pmb{y_5}$ & $\pmb{y_6}$ & $\pmb{y_4}$\\
  $x_3$ & $y_3$ & $y_1$ & $y_2$ & $\pmb{y_3}$ & $\pmb{y_1}$ & $\pmb{y_2}$
        & $y_6$ & $y_4$ & $y_5$ & $\pmb{y_6}$ & $\pmb{y_4}$ & $\pmb{y_5}$\\
  $x_4$ & $y_4$ & $y_5$ & $y_6$ & $\pmb{y_4}$ & $\pmb{y_5}$ & $\pmb{y_6}$
        & $y_1$ & $y_2$ & $y_3$ & $\pmb{y_1}$ & $\pmb{y_2}$ & $\pmb{y_3}$\\
  $x_5$ & $y_5$ & $y_6$ & $y_4$ & $\pmb{y_5}$ & $\pmb{y_6}$ & $\pmb{y_4}$
        & $y_2$ & $y_3$ & $y_1$ & $\pmb{y_2}$ & $\pmb{y_3}$ & $\pmb{y_1}$\\
  $x_6$ & $y_6$ & $y_4$ & $y_5$ & $\pmb{y_6}$ & $\pmb{y_4}$ & $\pmb{y_5}$
        & $y_3$ & $y_1$ & $y_2$ & $\pmb{y_3}$ & $\pmb{y_1}$ & $\pmb{y_2}$\\
\hline
  $y_1$ & $\pmb{x_1}$ & $\pmb{x_3}$ & $\pmb{x_2}$ & $x_1$ & $x_3$ & $x_2$
        & $\pmb{x_4}$ & $\pmb{x_6}$ & $\pmb{x_5}$ & $x_4$ & $x_6$ & $x_5$\\
  $y_2$ & $\pmb{x_2}$ & $\pmb{x_1}$ & $\pmb{x_3}$ & $x_2$ & $x_1$ & $x_3$
        & $\pmb{x_5}$ & $\pmb{x_4}$ & $\pmb{x_6}$ & $x_5$ & $x_4$ & $x_6$\\
  $y_3$ & $\pmb{x_3}$ & $\pmb{x_2}$ & $\pmb{x_1}$ & $x_3$ & $x_2$ & $x_1$
        & $\pmb{x_6}$ & $\pmb{x_5}$ & $\pmb{x_4}$ & $x_6$ & $x_5$ & $x_4$\\

  $y_4$ & $\pmb{x_4}$ & $\pmb{x_6}$ & $\pmb{x_5}$ & $x_4$ & $x_6$ & $x_5$
        & $\pmb{x_1}$ & $\pmb{x_3}$ & $\pmb{x_2}$ & $x_1$ & $x_3$ & $x_2$\\
  $y_5$ & $\pmb{x_5}$ & $\pmb{x_4}$ & $\pmb{x_6}$ & $x_5$ & $x_4$ & $x_6$
        & $\pmb{x_2}$ & $\pmb{x_1}$ & $\pmb{x_3}$ & $x_2$ & $x_1$ & $x_3$\\
  $y_6$ & $\pmb{x_6}$ & $\pmb{x_5}$ & $\pmb{x_4}$ & $x_6$ & $x_5$ & $x_4$
        & $\pmb{x_3}$ & $\pmb{x_2}$ & $\pmb{x_1}$ & $x_3$ & $x_2$ & $x_1$\\

\end{tabular}
\caption{Extending the normal super-game $U_i\rightarrow U_j$ between super-teams $u_i$ and $U_j$, where $U_i=u_{2i-1}\cup u_{2i}=\{x_1,x_2,x_3\}\cup\{x_4,x_5,x_6\}$, $U_j=u_{2j-1}\cup u_{2j}=\{y_1,y_2,y_3\}\cup\{y_4,y_5,y_6\}$, and home games are marked in bold}
\label{expand_normal}
\end{table}

\textbf{Left super-games}:
Assume that there is a left super-game $U_i\xrightarrow{L} U_j$, i.e., a left super-game between super-teams $U_i$ and $U_j$ at the home of $U_j$. Recall that $U_{i}$ represents cycle-teams $u_{2i-1}\cup u_{2i}$ and $U_{j}$ represents cycle-teams $u_{2j-1}\cup u_{2j}$. The left super-game will be extended to cycle-games in two sub-time slots: in the first sub-time slot, there are two \emph{left cycle-games}: $\{u_{2i-1}\xrightarrow{L} u_{2j-1},\  u_{2i}\xrightarrow{L} u_{2j}\}$, where $u_{i'}\xrightarrow{L} u_{i''}$ means a left cycle-game between cycle-teams $u_{i'}$ and $u_{i''}$ at the home of $u_{i''}$; in the second sub-time slot, there are still two left cycle-games: $\{u_{2i-1}\xrightarrow{L} u_{2j},\  u_{2i}\xrightarrow{L} u_{2j-1}\}$.

Next, we only need to explain the design of left cycle-games.

Assume that there is a left cycle-game $u_i\xrightarrow{L} u_j$, i.e., a left cycle-game between cycle-teams $u_i$ and $u_j$ at the home of $u_j$. Let $u_i=\{x_1,\dots, x_k\}$ and $u_j=\{y_1,\dots, y_k\}$.
Let $p_l$ denote the $k$ normal games in $\{x_{i}\rightarrow y_{(i+l-2)\bmod k)+1}\}_{i=1}^{k}$.
The left cycle-game will be extended to normal games on $2k$ days, which can
can be presented by
\[
\lrA{p_1\cdots p_{k-2}\cdot p_{k-1}\cdot\overline{p_k}}\cdot\lrA{\overline{p_1\cdots p_{k-2}\cdot p_{k-1}\cdot\overline{p_k}}}.
\]

An example of a left super-game with $k=3$ is shown in Table~\ref{expand_left}.
We can see that each left super-game will be extended to normal games on $4k$ days. There are at most $k$-consecutive home/away normal games. The normal games between normal teams in $U_i$ and normal teams in $U_j$, which can be seen as a directed complete bipartite graph $B(U_i,U_j)$, are all arranged.

\begin{table}[ht]
\centering
\begin{tabular}{c|cccccccccccc}
  & $1$ & $2$ & $3$ & $4$ & $5$ & $6$ & $7$ & $8$ & $9$ & $10$ & $11$ & $12$\\
\hline
  $x_1$ & $y_1$ & $y_2$ & $\pmb{y_3}$ & $\pmb{y_1}$ & $\pmb{y_2}$ & $y_3$
        & $y_4$ & $y_5$ & $\pmb{y_6}$ & $\pmb{y_4}$ & $\pmb{y_5}$ & $y_6$\\
  $x_2$ & $y_2$ & $y_3$ & $\pmb{y_1}$ & $\pmb{y_2}$ & $\pmb{y_3}$ & $y_1$
        & $y_5$ & $y_6$ & $\pmb{y_4}$ & $\pmb{y_5}$ & $\pmb{y_6}$ & $y_4$\\
  $x_3$ & $y_3$ & $y_1$ & $\pmb{y_2}$ & $\pmb{y_3}$ & $\pmb{y_1}$ & $y_2$
        & $y_6$ & $y_4$ & $\pmb{y_5}$ & $\pmb{y_6}$ & $\pmb{y_4}$ & $y_5$\\
  $x_4$ & $y_4$ & $y_5$ & $\pmb{y_6}$ & $\pmb{y_4}$ & $\pmb{y_5}$ & $y_6$
        & $y_1$ & $y_2$ & $\pmb{y_3}$ & $\pmb{y_1}$ & $\pmb{y_2}$ & $y_3$\\
  $x_5$ & $y_5$ & $y_6$ & $\pmb{y_4}$ & $\pmb{y_5}$ & $\pmb{y_6}$ & $y_4$
        & $y_2$ & $y_3$ & $\pmb{y_1}$ & $\pmb{y_2}$ & $\pmb{y_3}$ & $y_1$\\
  $x_6$ & $y_6$ & $y_4$ & $\pmb{y_5}$ & $\pmb{y_6}$ & $\pmb{y_4}$ & $y_5$
        & $y_3$ & $y_1$ & $\pmb{y_2}$ & $\pmb{y_3}$ & $\pmb{y_1}$ & $y_2$\\
\hline
  $y_1$ & $\pmb{x_1}$ & $\pmb{x_3}$ & $x_2$ & $x_1$ & $x_3$ & $\pmb{x_2}$
        & $\pmb{x_4}$ & $\pmb{x_6}$ & $x_5$ & $x_4$ & $x_6$ & $\pmb{x_5}$\\
  $y_2$ & $\pmb{x_2}$ & $\pmb{x_1}$ & $x_3$ & $x_2$ & $x_1$ & $\pmb{x_3}$
        & $\pmb{x_5}$ & $\pmb{x_4}$ & $x_6$ & $x_5$ & $x_4$ & $\pmb{x_6}$\\
  $y_3$ & $\pmb{x_3}$ & $\pmb{x_2}$ & $x_1$ & $x_3$ & $x_2$ & $\pmb{x_1}$
        & $\pmb{x_6}$ & $\pmb{x_5}$ & $x_4$ & $x_6$ & $x_5$ & $\pmb{x_4}$\\

  $y_4$ & $\pmb{x_4}$ & $\pmb{x_6}$ & $x_5$ & $x_4$ & $x_6$ & $\pmb{x_5}$
        & $\pmb{x_1}$ & $\pmb{x_3}$ & $x_2$ & $x_1$ & $x_3$ & $\pmb{x_2}$\\
  $y_5$ & $\pmb{x_5}$ & $\pmb{x_4}$ & $x_6$ & $x_5$ & $x_4$ & $\pmb{x_6}$
        & $\pmb{x_2}$ & $\pmb{x_1}$ & $x_3$ & $x_2$ & $x_1$ & $\pmb{x_3}$\\
  $y_6$ & $\pmb{x_6}$ & $\pmb{x_5}$ & $x_4$ & $x_6$ & $x_5$ & $\pmb{x_4}$
        & $\pmb{x_3}$ & $\pmb{x_2}$ & $x_1$ & $x_3$ & $x_2$ & $\pmb{x_1}$\\
\end{tabular}
\caption{Extending the left super-game $U_i\xrightarrow{L} U_j$ between super-teams $u_i$ and $U_j$, where $U_i=u_{2i-1}\cup u_{2i}=\{x_1,x_2,x_3\}\cup\{x_4,x_5,x_6\}$, $U_j=u_{2j-1}\cup u_{2j}=\{y_1,y_2,y_3\}\cup\{y_4,y_5,y_6\}$, and home games are marked in bold}
\label{expand_left}
\end{table}

\textbf{Right super-games}:
Assume that there is a right super-game $\{U_i\xrightarrow{R} U_j\}$, i.e., a right super-game between super-teams $U_i$ and $U_j$ at the home of $U_j$. There also should be a team-pair $R_{i'}$. Recall that $U_{i}$ represents cycle-teams $u_{2i-1}\cup u_{2i}$, $U_{j}$ represents cycle-teams $u_{2j-1}\cup u_{2j}$, and $R_{i}=\{t_{n_S+2i'-1}, t_{n_S+2i'}\}$. The super-game will be directly extended to normal-games on $4k$ days.
Let $U_i=\{x_1,\dots, x_k\}\cup\{x_{k+1},\dots,x_{2k}\}$ and $U_j=\{y_1,\dots, y_k\}\cup\{y_{k+1},\dots,y_{2k}\}$. In the schedule, if the position of $U_i$ is above $U_j$ (for example, the position of $U_4$ is above $U_5$ while the position of $U_6$ is below $U_3$ in Figure~\ref{figa}), we let $t_{n_S+2i'-1}=x_{2k+1}$ and $t_{n_S+2i'}=y_{2k+1}$, otherwise, we let $t_{n_S+2i'-1}=y_{2k+1}$ and $t_{n_S+2i'}=x_{2k+1}$. 
Moreover, let $p_l$ denote the $2k+1$ normal games in $\{x_{i}\rightarrow y_{(2k+2l-i)\bmod (2k+1)+1}\}_{i=1}^{2k+1}$.
The right super-game will be extended to normal games on $2k$ days, which can
can be presented by
\begin{align*}
&\lrA{p_1\cdots p_{k-2}\cdot \overline{p_{k-1}}\cdot p_k}\cdot\lrA{\overline{p_1\cdots p_{k-2}\cdot \overline{p_{k-1}}\cdot p_k}}\\
&\cdot \lrA{p_{k+1}\cdots p_{2k-2}\cdot \overline{p_{2k-1}}\cdot p_{2k}}\cdot\lrA{\overline{p_{k+1}\cdots p_{2k-2}\cdot \overline{p_{2k-1}}\cdot p_{2k}}}.
\end{align*}

An example of a right super-game with $k=3$ is shown in Table~\ref{expand_right}.
We can see that each right super-game will be extended to normal games on $4k$ days. There are at most $k$-consecutive home/away normal games. Comparing with normal/left super-games, there are two more normal teams in $R_{i'}$, and hence two days of normal games, $p_{2k+1}\cup\overline{p_{2k+1}}$, are not arranged. We will arrange them after the super-games in the last time slot.

\begin{table}[ht]
\centering
\begin{tabular}{c|cccccccccccc}
  & $1$ & $2$ & $3$ & $4$ & $5$ & $6$ & $7$ & $8$ & $9$ & $10$ & $11$ & $12$\\
\hline
  $x_1$ & $y_1$ & $\pmb{y_3}$ & $y_5$ & $\pmb{y_1}$ & $y_3$ & $\pmb{y_5}$ & $y_7$ & $\pmb{y_2}$ & $y_4$ & $\pmb{y_7}$ & $y_2$ & $\pmb{y_4}$\\
  $x_2$ & $y_7$ & $\pmb{y_2}$ & $y_4$ & $\pmb{y_7}$ & $y_2$ & $\pmb{y_4}$ & $y_6$ & $\pmb{y_1}$ & $y_3$ & $\pmb{y_6}$ & $y_1$ & $\pmb{y_3}$\\
  $x_3$ & $y_6$ & $\pmb{y_1}$ & $y_3$ & $\pmb{y_6}$ & $y_1$ & $\pmb{y_3}$ & $y_5$ & $\pmb{y_7}$ & $y_2$ & $\pmb{y_5}$ & $y_7$ & $\pmb{y_2}$\\
  $x_4$ & $y_5$ & $\pmb{y_7}$ & $y_2$ & $\pmb{y_5}$ & $y_7$ & $\pmb{y_2}$ & $y_4$ & $\pmb{y_6}$ & $y_1$ & $\pmb{y_4}$ & $y_6$ & $\pmb{y_1}$\\
  $x_5$ & $y_4$ & $\pmb{y_6}$ & $y_1$ & $\pmb{y_4}$ & $y_6$ & $\pmb{y_1}$ & $y_3$ & $\pmb{y_5}$ & $y_7$ & $\pmb{y_3}$ & $y_5$ & $\pmb{y_7}$\\
  $x_6$ & $y_3$ & $\pmb{y_5}$ & $y_7$ & $\pmb{y_3}$ & $y_5$ & $\pmb{y_7}$ & $y_2$ & $\pmb{y_4}$ & $y_6$ & $\pmb{y_2}$ & $y_4$ & $\pmb{y_6}$\\
  $x_7$ & $y_2$ & $\pmb{y_4}$ & $y_6$ & $\pmb{y_2}$ & $y_4$ & $\pmb{y_6}$ & $y_1$ & $\pmb{y_3}$ & $y_5$ & $\pmb{y_1}$ & $y_3$ & $\pmb{y_5}$\\
\hline
  $y_1$ & $\pmb{x_1}$ & $x_3$ & $\pmb{x_5}$ & $x_1$ & $\pmb{x_3}$ & $x_5$ & $\pmb{x_7}$ & $x_2$ & $\pmb{x_4}$ & $x_7$ & $\pmb{x_2}$ & $x_4$\\
  $y_2$ & $\pmb{x_7}$ & $x_2$ & $\pmb{x_4}$ & $x_7$ & $\pmb{x_2}$ & $x_4$ & $\pmb{x_6}$ & $x_1$ & $\pmb{x_3}$ & $x_6$ & $\pmb{x_1}$ & $x_3$\\
  $y_3$ & $\pmb{x_6}$ & $x_1$ & $\pmb{x_3}$ & $x_6$ & $\pmb{x_1}$ & $x_3$ & $\pmb{x_5}$ & $x_7$ & $\pmb{x_2}$ & $x_5$ & $\pmb{x_7}$ & $x_2$\\
  $y_4$ & $\pmb{x_5}$ & $x_7$ & $\pmb{x_2}$ & $x_5$ & $\pmb{x_7}$ & $x_2$ & $\pmb{x_4}$ & $x_6$ & $\pmb{x_1}$ & $x_4$ & $\pmb{x_6}$ & $x_1$\\
  $y_5$ & $\pmb{x_4}$ & $x_6$ & $\pmb{x_1}$ & $x_4$ & $\pmb{x_6}$ & $x_1$ & $\pmb{x_3}$ & $x_5$ & $\pmb{x_7}$ & $x_3$ & $\pmb{x_5}$ & $x_7$\\
  $y_6$ & $\pmb{x_3}$ & $x_5$ & $\pmb{x_7}$ & $x_3$ & $\pmb{x_5}$ & $x_7$ & $\pmb{x_2}$ & $x_4$ & $\pmb{x_6}$ & $x_2$ & $\pmb{x_4}$ & $x_6$\\
  $y_7$ & $\pmb{x_2}$ & $x_4$ & $\pmb{x_6}$ & $x_2$ & $\pmb{x_4}$ & $x_6$ & $\pmb{x_1}$ & $x_3$ & $\pmb{x_5}$ & $x_1$ & $\pmb{x_3}$ & $x_5$\\
\end{tabular}
\caption{Extending the right super-game $U_i\xrightarrow{R} U_j$ between super-teams $u_i$ and $U_j$ with a team-pair $R_{i'}$, where $U_i=u_{2i-1}\cup u_{2i}=\{x_1,x_2,x_3\}\cup\{x_4,x_5,x_6\}$, $U_j=u_{2j-1}\cup u_{2j}=\{y_1,y_2,y_3\}\cup\{y_4,y_5,y_6\}$, $R_{i'}=\{x_{2k+1},y_{2k+1}\}$, and home games are marked in bold}
\label{expand_right}
\end{table}

The first $m-2$ time slots will be extended to $4k(m-2)$ days. Recall that $n=2km+2r$. Each normal team will have $2n-2-4k(m-2)=8k+4r-2$ remaining normal games, which correspond to the normal games in the last time slot.

\textbf{The last time slot.}
In the last time slot, we will first design $m/2$ super-games, and then we will design the unarranged normal games.
Figure~\ref{figc} shows the $m/2$ super-games of our schedule in the last time slot.
Except for the $r$ right super-games, the rest super-games are all left super-games.

\begin{figure}[ht]
    \centering
    \includegraphics[scale=0.7]{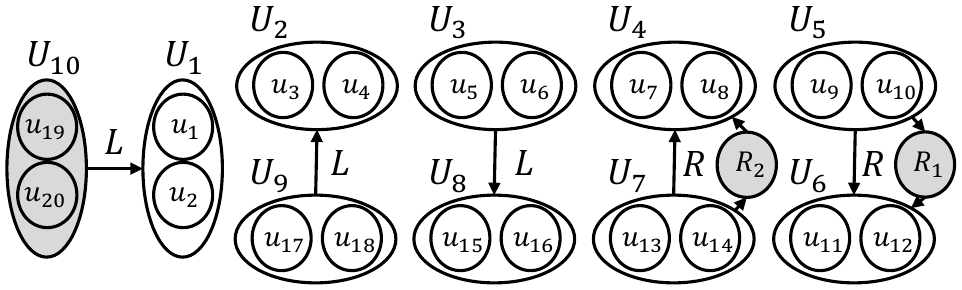}
    \caption{The schedule in the last time slot for an instance with $m=10$ and $r=2$}
    \label{figc}
\end{figure}

We use $H$ and $A$ to denote a home and an away normal game, respectively. An advantage of left/right super-games is that: by the design of left/right super-games as shown in Tables~\ref{expand_left} and~\ref{expand_right}, all normal teams play $AH$ or $HA$ on the last two days of the super-games. It can help us to combine the unarranged normal games without creating more than $k$-consecutive home/away normal games.

Next, we will design the unarranged normal games. Since the $m/2$ super-games in the last time slot have already spanned $4k$ days, the unarranged normal games will span $4k+4r-2$ days.
To design unarranged normal games, we will consider two cases of normal teams.

\textbf{Case~1:} We first consider the normal teams in super-team $U_{m}$ and the $r$ team-pairs $\{R_1,\dots,R_r\}$.
By the design of left, normal, and right super-games, we have not arranged the normal games between normal teams in each super-team and team-pair.
Note that the positions of super-team $U_{m}$ and the $r$ team-pairs are fixed, and hence they never meet each other in the schedule. The normal games between normal teams in $U_m$ and normal teams in $R_i$, and the normal games between normal teams in $R_{i}$ and normal teams in $R_{i'}$, are unarranged.
There are $2k$ normal teams in $U_m$ and $2$ normal teams in each team-pair.
The normal games between normal teams in $U_m\cup R_1\cup\cdots R_r$, which can seen as a directed complete graph, form a double round-robin on these $2k+2r$ (even) normal teams.
Hence, we can take $2\cdot(2k+2r-1)=4k+4r-2$ days to arrange a double-round robin for them. That is exactly the number of left days.
To design a double round-robin, we need to make sure that it can combine well with the super-games in the last time-slot, i.e., without creating more than $k$-consecutive home/away normal games.
Since all of them play $AH$ or $HA$ on the last two days of the super-games in the last time slot, we can call an algorithm of TTP-2 to satisfy the bounded-by-$k$ and no-repeat constraints~\cite{thielen2012approximation,DBLP:conf/cocoon/ZhaoX21}. Since there are at most two consecutive home/away normal games for TTP-2, we can put the TTP-2 double-round robin behind the super-games in the last time slot, which will not create more than three consecutive home/away normal games.

\textbf{Case~2:} Next, we consider the normal teams in super-teams $U_1\cup\cdots U_{m-1}$.
The unarranged normal games can be summarized as two parts. Recall that there are two days of normal games $p_{2k+1}\cup\overline{p_{2k+1}}$ unarranged in each right super-game. These unarranged normal games are regarded as the first part. The normal games between normal teams in each super-team are regarded as the second part.

\textbf{The first part.}
Since there are $r$ team-pairs $\{R_1,\dots,R_r\}$, we have $r$ kinds of right super-games. The right super-game including team-pair $R_{i'}$ is called the $i'$-th right super-game.
We use an super-edge $(U_i, U_j)_{i'}$ to denote an $i'$-th right super-game between super-teams $U_i$ and $U_j$.
In our schedule, each super-team has exactly two $i'$-th right super-games. Hence, the super-edges representing the $i'$-th right super-games is a set of cycles, denoted by $\C_{<i'>}$. Whether there exists an $i'$-th right super-game between super-teams $U_i$ and $U_j$ is due to their relative positions. It is easy to see that $\C_{<i'>}$ is a cycle packing, i.e., each cycle in  $\C_{<i'>}$ has the same length. Note that there are $m-1$ super-teams in $\C_{<i'>}$. Then, the length of each cycle in $\C_{<i'>}$ is $(m-1)/\size{\C_{<i'>}}$, which is odd since $m-1$ is odd. Hence, $\C_{<i'>}$ is an odd cycle packing.
For example, when $m=10$ and $r=2$ as shown in Figure~\ref{figa}, $\C_{<1>}$ contains one (odd) cycle $(U_1,U_2,\dots, U_9, U_1)_1$, and $\C_{<2>}$ contains three (odd) cycles $(U_1,U_4,U_7,U_1)_2$, $(U_2,U_5,U_8,U_2)_2$ and $(U_3,U_6,U_9,U_3)_2$.

Then, we analyze the unarranged normal games in the $i'$-right super-games. Take an arbitrary odd cycle in $\C_{<i'>}$, denoted by $C=(U_{i'_1},U_{i'_2},\dots, U_{i'_p}, U_{i'_1})_{i'}$, where $p$ is odd.
Let $U_{i'_i}=\{x^i_1,\dots, x^i_k\}\cup\{x^i_{k+1},\dots,x^i_{2k}\}$.
Assume that there is an $i'$-th right super-game between super-teams $U_{i'_i}$ and $U_{i'_j}$, and let the team-pair $R_{i'}=\{x^i_{2k+1},x^j_{2k+1}\}$. By the design of right super-games, the unarranged normal games are $p_{2k+1}\cup\overline{p_{2k+1}}=\{x^i_{i''}\leftrightarrow x^j_{(2k-i'')\bmod (2k+1)+1}\}_{i''=1}^{2k}\cup\{x^i_{2k+1}\leftrightarrow x^j_{2k+1}\}$. Note that the normal games $\{x^i_{2k+1}\leftrightarrow x^j_{2k+1}\}$ between normal teams in team-pair $R_{i'}$ have been arranged when we consider the super-team $U_m$ and the $r$ team-pairs in {Case~1}. Hence, we need to arrange the normal games $\{x^i_{i''}\leftrightarrow x^j_{(2k-i'')\bmod (2k+1)+1}\}_{i''=1}^{2k}$. An example of these unarranged normal games for $k=3$ is shown in Figure~\ref{figr1}.

\begin{figure}[ht]
    \centering
    \includegraphics[scale=1.2]{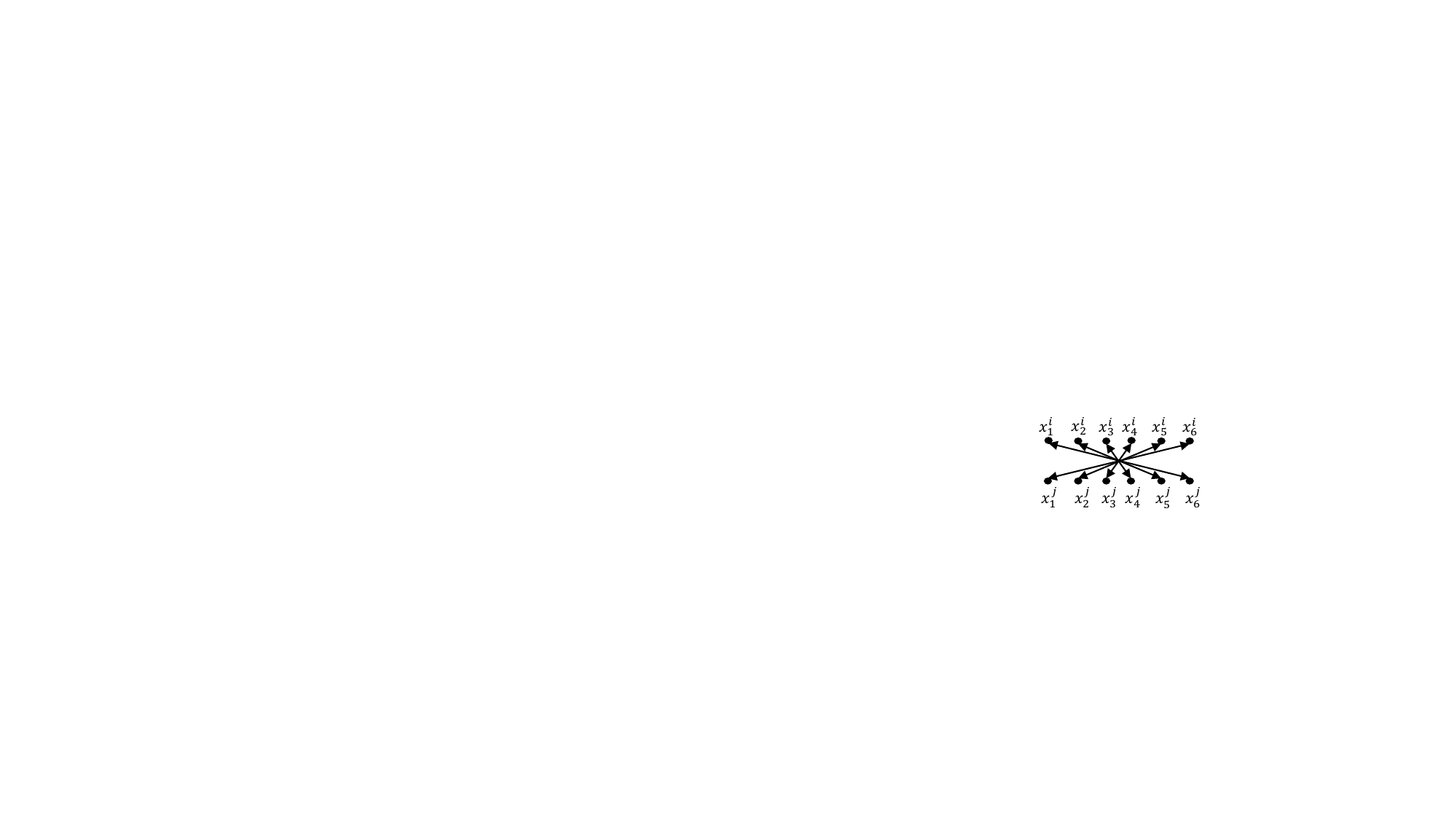}
    \caption{The unarranged normal games between super-teams $U_i=\{x^i_1,x^i_2,x^i_3\}\cup\{x^i_4,x^i_5,x^i_6\}$, $U_j=\{x^j_1,x^j_2,x^j_3\}\cup\{x^j_4,x^j_5,x^j_6\}$}
    \label{figr1}
\end{figure}

We can see that the structure has a symmetric property. In the following, we only consider two normal teams $\{x^i_1,x^i_{2k}\}$ for each super-team $U_{i'_i}$. The unarranged normal games on the cycle $C=(U_{i'_1},U_{i'_2},\dots, U_{i'_p}, U_{i'_1})_{i'}$ can be presented by a bi-directed cycle
\[
x^1_0\leftrightarrow x^2_{2k}\leftrightarrow x^3_0\cdots\leftrightarrow x^{p}_{0}\leftrightarrow x^1_{2k}\leftrightarrow x^2_0\leftrightarrow x^3_{2k}\cdots x^{p}_{2k}\leftrightarrow x^1_0.
\]
An example of these $4p$ unarranged normal games for $k=3$ and $p=9$ is shown in Figure~\ref{figr2}.

\begin{figure}[ht]
    \centering
    \includegraphics[scale=1.2]{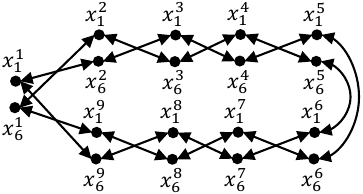}
    \caption{The unarranged normal games on the cycle $C=(U_{i'_1},U_{i'_2},\dots, U_{i'_p}, U_{i'_1})_{i'}$, where we only consider two normal teams $\{x^i_1,x^i_{2k}\}$ for each super-team $U_{i'_i}$, $p=9$, and $k=3$}
    \label{figr2}
\end{figure}
Let $s_1$ and $s_2$ be the $p$ normal games $\{x^i_1\rightarrow x^{(i+p)\bmod p+1}_{2k}\}_{i=1}^{i=p}$ and $\{x^i_{2k}\rightarrow x^{(i+p)\bmod p+1}_{1}\}_{i=1}^{i=p}$, respectively. Then, the $4p$ unarranged normal games can be arranged as $(s_1\cdot s_2)\cdot(\overline{s_1\cdot s_2})$.
An example of the $p$ normal games of $s_1$ and $s_2$ for $k=3$ and $p=9$ is shown in Figure~\ref{figr3}. Note that we can design the unarranged normal games over the cycles in $\C_{<i'>}$ simultaneously. Moreover, all normal teams in $U_1\cup U_2\cup \cdots U_{m-1}$ play $AHHA$ or $HAAH$.

\begin{figure}[ht]
    \centering
    \includegraphics[scale=1.2]{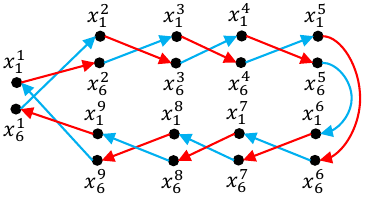}
    \caption{The unarranged normal games on the cycle $C=(U_{i'_1},U_{i'_2},\dots, U_{i'_p}, U_{i'_1})_{i'}$, where we only consider two normal teams $\{x^i_1,x^i_{2k}\}$ for each super-team $U_{i'_i}$, $p=9$, and $k=3$}
    \label{figr3}
\end{figure}

Since there are $r$ kinds of right super-games, we can design the unarranged normal games related to the cycle packing $\C_{<i'>}$, respectively, and hence the unarranged normal games in the first part need to span $4r$ days. Recall that all normal teams play $HA$ or $AH$ on the last two days of the super-games in the last time slot. Putting the first part behind the super-games will not create more than three consecutive home/away normal games.

\textbf{The second part.} Recall that the second part contains the normal games between normal teams in each super-team $U_i$ ($i<m$). For each super-team $U_i$, there are $2k$ (even) normal teams, and the unarranged normal games form a double round robin for them. Similarly, we can call an algorithm of TTP-2 to arrange a double round-robin for them. The second part spans $4k-2$ days.

For normal teams in $U_1\cup U_2\cup \cdots U_{m-1}$, the normal games in these two parts span $4k+4r-2$ days, which is exactly the number of remaining days.
Note that all normal teams in $U_1\cup U_2\cup \cdots U_{m-1}$ play $AH$ or $HA$ on the last two days of the first part. Hence, putting the TTP-2 double round-robins behind the first part will not create more than three consecutive home/away normal games.

Note that there are at most $k$-consecutive home/away normal games in the super-games of the last time slot. Putting the unarranged normal games behind the super-games will not create more than three consecutive home/away normal games.

\begin{theorem}~\label{feasible1}
For TTP-$k$ with any $k\geq 3$, when $n\geq 8k^2$, the $k$-cycle packing construction can generate a feasible schedule.
\end{theorem}
\begin{proof}
First, we show that our schedule is a complete double round-robin.

We use $(X,Y)$ to denote the normal games between normal teams in $X$ and normal teams in $Y$, and $(X)$ to denote the normal games between normal teams in $X$, where $X,Y\subseteq \{t_1,\dots,t_n\}$ and $X\cap Y=\emptyset$.
For each team-pair $R_{i'}$, it plays two right super-games with each super-team $U_i\in\{U_1,\dots,U_{m-1}\}$. Then, the normal games $(R_{i'},U_i)$ are arranged.
For super-teams $U_i$ and $U_j$, there is a normal/left/right super-game between them, and then the normal games $(U_i,U_j)$ are arranged (for the right super-game, the two days of unarranged normal games are designed in Case~2).
The normal games $(R_1\cup\cdots R_{r}\cup U_m)$
are arranged in Case~1. Hence, the schedule is complete.

Next, we show that the no-repeat property holds.

It is easy to see that the two normal games between any pair of normal teams are designed in a super-game or a TTP-2 double round-robin (see Case~1 and the second part in Case~2) or the first part in Case~2. By the design of super-games and the first part in Case~2, it is easy to see that there are no-repeat normal games. For a TTP-2 double round-robin, it is clear that the normal games satisfy the no-repeat property. So, the no-repeat property holds.

Then, we prove the bounded-by-$k$ property.

Recall that putting the unarranged normal games behind the super-games in the last time slot will not create more than three consecutive home/away normal games. Hence, we can assume that the schedule in the last time slot only contains the $m/2$ super-games.
For super-team $U_m$, it plays away and home left super-games alternatively.
For super-team $U_i$ ($i<m$), it always play away (resp., home) normal/right super-games until it plays the first away (resp., a home) left super-game. After playing the first away (resp., home) left super-game, it will always play home (resp., away) normal/left/right super-games. Note that a super-team can have at most two left super-games, and the second left super-game must appear in the last time slot.
For each team-pair $R_{i'}$, there are two normal teams $\{t_{n_S+2i'-1}, t_{n_S+2i'}\}$. In the right super-games, the normal team $t_{n_S+2i'-1}$ always belongs to the upper super-team, while the normal team $t_{n_S+2i'}$ always belongs to the lower super-team. Hence, both of them play home and away right super-games alternatively.
For the sake of presentation, we use $H^{l}$ (resp., $A^{l}$) to denote $l$-consecutive home (resp., away) normal games.
If a super-team plays an away (resp., a home) normal super-game, we know that the normal teams in it will play $A^{k}\cdot H^k\cdot A^k\cdot H^k$ (resp., $H^{k}\cdot A^k\cdot H^k\cdot A^k$).
If a super-team plays an away (resp., a home) left super-game, we know that the normal teams in it will play $A^{k-1}\cdot H^k\cdot A^k\cdot H^k\cdot A$ (resp., $H^{k-1}\cdot A^k\cdot H^k\cdot A^k\cdot H$).
If a super-team plays an away (resp., a home) right super-game, we know that the normal teams in it will play $(A^{k-2}\cdot H\cdot A\cdot H^{k-2} \cdot A \cdot H)\cdot (A^{k-2}\cdot H\cdot A\cdot H^{k-2} \cdot A \cdot H)$ (resp., $(H^{k-2}\cdot A\cdot H\cdot A^{k-2} \cdot H \cdot A)\cdot (H^{k-2}\cdot A\cdot H\cdot A^{k-2} \cdot H \cdot A)$). We can see that there are no more than $k$-consecutive home/away normal games.
\end{proof}

\subsubsection{Analyzing the quality of the construction}
Recall the three complete graphs $H=(S,E_H)$, $\G=(\V,\E)$, and $\H=(\U,\F)$, where $S=\{t_1,\dots,t_{n_S}\}$ representing the $n_S$ normal teams,
$\V=\{u_1,\dots,u_{2m}\}$ representing the $2m$ cycle-teams, and $\U=\{U_1,\dots,U_m\}$ representing the $m$ super-teams.

Given a $k$-cycle packing $\C_k=\{C_1,\dots,C_{2m}\}$, where the indies are arbitrarily labeled. The construction can generate a feasible solution, but it may not be good enough. Hence, to get a schedule with a small total traveling distance, we need to find a good way to label the cycles. To do this, we use the random method: cycles are labeled using $\{C_1,\dots,C_{2m}\}$ uniformly at random. Note that the $n_{\overline{S}}$ normal teams in $\overline{S}$ are still divided into $r$ team-pairs arbitrarily.
We will calculate the expected weight of our schedule.

Note that $\E$ contains $\frac{2m(2m-1)}{2}$ edges, where $\frac{2m(2m-1)}{2}\geq m^2$ since $m=2\floor{\frac{n}{4k}}>1$.
For an edge $(u_i,u_j)\in\E$, by (\ref{eqn_GH}), we have that
\begin{equation}\label{eq1}
\EE{w(u_i,u_j)}=\frac{2}{2m(2m-1)}w(\E)\leq \frac{1}{m^2}w(\E)=\frac{1}{2m^2}\Delta_H.
\end{equation}
Recall that $w(U_i,U_j)=\sum_{u_{i'}\in U_i\&u_{j'}\in U_j}w(u_{i'},u_{j'})$ for $i\neq j$. For an edge $(U_i,U_j)\in \F$, we have that
\begin{equation}\label{eq2}
\EE{w(U_i,U_j)}=\sum_{u_{i'}\in U_i\&u_{j'}\in U_j}\EE{w(u_{i'},u_{j'})}\leq\sum_{u_{i'}\in U_i\&u_{j'}\in U_j}\frac{1}{2m^2}\Delta_H=\frac{2}{m^2}\Delta_H.
\end{equation}

\begin{lemma}\label{extra}
Assume that there is a normal cycle-game between cycle-teams $u_i$ and $u_j$. There is an $O(k^2)$-time algorithm ensuring that the total weight of all normal games in the cycle game is at most $(4/k)w(u_i,u_j)+(k-1)(w(C_i)+w(C_j))$.
\end{lemma}
\begin{proof}
Assume w.l.o.g. that the normal cycle-game is at the home of $u_j$. By the design of normal cycle-games, all normal teams in $u_i$ (resp., $u_j$) play $A^k\cdot H^k$ (resp., $H^k\cdot A^k$). Hence, each normal team in $u_i\cup u_j$ has exactly one road trip visiting $k$ normal teams. Recall that $u_i=\{t_{k(i-1)+1}, \dots, t_{k(i-1)+k}\}$ and $u_j=\{t_{k(j-1)+1}, \dots, t_{k(j-1)+k}\}$. For the sake of presentation, we let $t_{k(i-1)+i'}=x_{i'}$ and $t_{k(j-1)+i'}=y_{i'}$, i.e., $u_i=\{x_1,\dots, x_k\}$ and $u_j=\{y_1,\dots, y_k\}$.
By the design of normal cycle-games, the weight of these $2k$ road trips are
\begin{align*}
\sum_{i'=1}^k(w(x_{i'},y_{i'})+w(x_{i'},y_{(i'+k-2)\bmod k+1}))+(k-1)(w(C_i)+w(C_j)).
\end{align*}

Before extending the normal cycle-game into normal games, we may change the labels of normal teams in the $k$-cycle $C_j$ by setting $y^l_{i'}=y_{(i'+l-2)\bmod k+1}$, where $l\in\{1,2,\dots,k\}$. Note that after extending the normal cycle-game, we can recover the previous labels of normal teams in $C_j$, which does not effect the feasibility of teams since all normal teams in $u_i$ (resp., $u_j$) play $A^k\cdot H^k$ (resp., $H^k\cdot A^k$). Note that
\begin{align*}
&\sum_{l=1}^{k}\sum_{i'=1}^k(w(x_{i'},y^l_{i'})+w(x_{i'},y^l_{(i'+k-2)\bmod k+1}))=2w(u_i,u_j).
\end{align*}
Hence, for all $k$ cases of $l$, there must be one case $l=l_0$ such that
\[
\sum_{i'=1}^k(w(x_{i'},y^{l_0}_{i'})+w(x_{i'},y^{l_0}_{(i'+k-2)\bmod k+1}))\leq (4/k)w(u_i,u_j).
\]
The best label of the $k$ cases can be found in $O(k^2)$ time.
\end{proof}

Note that the algorithm in Lemma~\ref{extra} runs in constant time if $k=O(1)$.

Next, we are ready to analyze the total weight of schedule.
For the sake of analysis, we make the following assumption.
\begin{assumption}\label{AS1}
We assume that all normal teams return home before and after every normal game in left/right super-games and the unarranged normal games in the last time slot.
\end{assumption}

We will consider the following five parts of our schedule:
\begin{itemize}
\item $W_a$: the weight of normal games involving normal teams in $R_1\cup\cdots R_r$;
\item $W_b$: the weight of normal games between normal teams in each super-team of $\{U_1,\dots,U_m\}$;
\item $W_c$: the weight of normal games in left super-games;
\item $W_d$: the weight of the rest normal games in right super-games with the unarranged normal games in each right super-game;
\item $W_e$: the weight of normal games in normal super-games.
\end{itemize}

\textbf{Case~1: $W_a$.} The normal games appear in right super-games and Case~1 of the unarranged normal games. Recall that $R_1\cup\cdots R_r=\overline{S}$. By Assumption~\ref{AS1}, the weight of normal games involving the normal team $v\in \overline{S}$ is $4\deg_G(v)$. Thus, we have that
\begin{equation}\label{W_a}
\EE{W_a}=W_a\leq \sum_{v\in \overline{S}}4\deg_G(v)\leq (16k/n)\Delta_G\leq (8k^2/n)\psi,
\end{equation}
where the second inequality follows from (\ref{degree+}) and the last follows from Lemma \ref{lb-delta+}.

\textbf{Case~2: $W_b$.} The normal games appear in the second part of Case~2 of the unarranged normal games. By Assumption~\ref{AS1}, the weight of normal games between normal teams in super-team $U_i\in\U$ is $4w(U_i)$. Recall that $U_i=u_{2i-1}\cup u_{2i}$.
By the definition, $w(U_i)=w(u_{2i-1})+w(u_{2i})+w(u_{2i-1},u_{2i})$. By the triangle inequality, we have $w(u_{2i-1})\leq w(u_{2i-1},u_{2i})$ and $w(u_{2i})\leq w(u_{2i-1},u_{2i})$. Hence, we have $w(U_i)\leq 3w(u_{2i-1},u_{2i})$. Thus, we have that
\begin{equation}\label{W_b}
\EE{W_b}=\sum_{U_i\in\U}4\EE{w(U_i)}\leq\sum_{U_i\in\U}12\EE{w(u_{2i-1},u_{2i})}\leq (6/m)\Delta_H\leq (12k^2/n)\psi,
\end{equation}
where the third inequality follows from (\ref{eq1}), and the last follows from $m=2\floor{\frac{n}{4k}}>2(\frac{n}{4k}-1)\geq \frac{n}{4k}$ (since $n\geq 8k^2$) and Lemma \ref{lb-delta}.

\textbf{Case~3: $W_c$.} The normal games appear in left super-games. Assume that there is a left super-game between super-teams $U_i$ and $U_j$. By Assumption~\ref{AS1}, the weight of normal games in the left super-game is $4w(U_i,U_j)$. By (\ref{eq2}), the expected weight is $4w(U_i,U_j)=(8/m^2)\Delta_H$. There are $(m-2)+(m/2-r)\leq (3/2)m$ left super-games in our schedule. Thus, we have that
\begin{equation}\label{W_c}
\EE{W_c}\leq (12/m)\Delta_H=(24k^2/n)\psi.
\end{equation}

\textbf{Case~4: $W_d$.} The normal games appear in right super-games and the first part of Case~1. There should be a team-pair. Note that the weight of normal games related to the team-pair has been analyzed in $W_a$. We only consider the rest normal games. Assume that there is a right super-game between super-teams $U_i$ and $U_j$. By Assumption~\ref{AS1}, the weight of the normal games in the right super-game with the unarranged normal games is $4w(U_i,U_j)$. The expected weight is $(8/m^2)\Delta_H$.
There are $(m-1)r\leq mr\leq 2mk$ right super-games in our schedule since $r<2k$ by our setting. Thus, we have that
\begin{equation}\label{W_d}
\EE{W_d}\leq (16k/m)\Delta_H=(32k^3/n)\psi.
\end{equation}

\textbf{Case~5: $W_e$.}
Assume that there is a normal super-game between super-teams $U_{i}$ and $U_{j}$.
There are four normal cycle-games. By Lemma~\ref{extra}, the total weight of these four normal cycle-games are $\sum_{u_{i'}\in U_{i}\& u_{j'}\in U_{j}}(\frac{4}{k}w(u_{i'},u_{j'}))+(k-1)(w(C_{i'})+w(C_{j'}))$.
It is easy to see that $\EE{w(C_{i'})}=\EE{w(C_{j'})}=\frac{1}{2m}w(\C_k)$.
By (\ref{eq1}) and $(\ref{eqn_GH})$, we have that $\EE{w(u_i,u_j)}=\frac{2}{2m(2m-1)}w(\E)\leq\frac{1}{2m(2m-1)}\Delta_H\leq\frac{1}{4m(m-1)}\Delta_H$.
Hence, the expected weight of the normal super-game is $\frac{4}{m(m-1)k}\Delta_H+\frac{4(k-1)}{m}w(\C_k)$.
Note that there are at most $m(m-1)/2$ normal super-games since there are $m$ super-teams and there is exactly one super-game between each pair super-team.
The total expected weight of all normal super-games are $\frac{2}{k}\Delta_H+2(k-1)(m-1)w(\C_k)\leq \frac{2}{k}\Delta_H+2(k-1)mw(\C_k)$.
Since $n_S=2mk$, we have that
\begin{equation}\label{W_e}
\EE{W_e}\leq (2/k)\Delta_H+(1-1/k)n_Sw(\C_k).
\end{equation}

Thus, the expected weight of our schedule is bounded by
\[
(2/k)\Delta_H+(1-1/k)n_Sw(\C_k)+(44k^2/n)\psi+(32k^3/n)\psi.
\]
Recall that our schedule requires $n\geq 8k^2$. If $n\leq 8k^2=O(1)$, we can find an optimal solution by brute force, which takes constant time. Hence, we can get the following theorem.
\begin{theorem}\label{cycle-packing}
Let $\C_k$ be a $k$-cycle packing in graph $H$. For TTP-$k$ with any constant $k\geq3$, there is a randomized polynomial-time algorithm that can generate a solution with a total weight of at most $\frac{k-1}{k}n_S w(\C_k)+\frac{2}{k}\Delta_H+O(k^3/n)\psi$.
\end{theorem}

Given a $k$-path packing $\P_k$ in graph, we can obtain a $k$-cycle packing $\C_k$ by completing the $k$-paths in the packing, such that $w(\C_k)\leq 2w(\P_k)$ by the triangle inequality. By Theorem~\ref{cycle-packing}, we can use this packing $\C_k$ to get a solution of TTP-$k$. Hence, we can get the following corollary.
\begin{corollary}\label{path-packing}
Let $\P_k$ be a $k$-path packing in graph $H$. For TTP-$k$ with any constant $k\geq3$, there is a randomized polynomial-time algorithm that can generate a solution with a total weight of at most $\frac{2k-2}{k}n_S w(\P_k)+\frac{2}{k}\Delta_H+O(k^3/n)\psi$.
\end{corollary}

Our algorithm is randomized since we label the $2m$ cycle-teams by $\{C_1,\dots,C_{2m}\}$ randomly.
It can be derandomized efficiently by the classic method of conditional expectations \cite{motwani1995randomized}.

Next, we give a brief introduction to the derandomization.

According to the previous analysis, it is easy to see that there is an upper bound of the weight of our schedule, denoted by $W$, such that it is a linear function on variables $w(u_i,u_j)$ and $w(C_i)$ for $1\leq i,j\leq 2m$.
In Theorem~\ref{cycle-packing}, we have proved that $\EE{W}\leq\frac{k-1}{k}n_S w(\C_k)+\frac{2}{k}\Delta_H+O(k^3/n)\psi$.
Recall that the algorithm takes a $k$-cycle packing $\C_k$ as a part of the input, where $\size{\C_k}=2m$.
Let the $2m$ cycle teams be $\{C^*_1,\dots,C^*_{2m}\}$.
The main idea is to find a permutation $\sigma: (1,2,\dots,2m)\leftrightarrow (\sigma_1, \sigma_2,\dots,\sigma_{2m})$ such that
\[
\EE{W|C_1=C^*_{\sigma_1},\dots,C_{2m}=C^*_{\sigma_{2m}}}\leq \EE{W}.
\]
We can determine each $\sigma_i$ sequentially. Assume that we have determined $(\sigma_1, \dots,\sigma_{s-1})$ such that
\[
\EE{W|C_1=C^*_{\sigma_1},\dots,C_{s-1}=C^*_{\sigma_{s-1}}}\leq \EE{W}.
\]
To determine $\sigma_{s}$, we can simply let $\sigma_s$ be
\begin{equation}
\sigma_s=\arg\min_{\sigma_s\in \{1,2,\dots,2m\}\setminus\{\sigma_1, \sigma_2,\dots,\sigma_{s-1}\}} \EE{W|C_1=C^*_{\sigma_1},\dots,C_s=C^*_{\sigma_s}}\leq \EE{W}.
\end{equation}
Then, we can get
\begin{equation}
\EE{W|C_1=C^*_{\sigma_1},\dots,C_s=C^*_{\sigma_s}}\leq \EE{W|C_1=C^*_{\sigma_1},\dots,C_{s-1}=C^*_{\sigma_{s-1}}}\leq \EE{W}.
\end{equation}
Therefore, we can repeat this procedure to determine the permutation $\sigma$.

Now, we consider the running time.
We need to determine $2m$ labels respectively. For each label, we need to compute the expected conditional weight of $W$ at most $2m$ times. Hence, we need to compute the expected conditional weight of $W$ at most $4m^2=O(n^2)$ times.
To compute the expected conditional weight of $W$, we need to first compute the expected conditional weight of each variable.
There are $\frac{2m(2m-1)}{2}+2m=2m^2+m=O(n^2)$ variables.
It is easy to see that the expected conditional weight of each variable can be computed in $O(n^2)$ time, and hence the expected conditional weight of $W$ can be computed in $O(n^4)$ time.
Therefore, the derandomization takes $O(n^6)$ extra time.

Next, we are ready to analyze the approximation qualities of these two constructions.
We will consider TTP-3, TTP-4, TTP-$k$ with $k\geq 5$, and LDTTP-$k$ with $k\geq 3$, respectively.

\section{The Approximation Ratio of TTP-$k$}\label{C}
\subsection{TTP-3}\label{C.1}
\subsubsection{The Approximation Quality of the Hamiltonian Cycle Construction}
Note that the Hamiltonian cycle they used is generated by the Christofides-Serdyukov algorithm. In our algorithm, we also consider another Hamiltonian cycle that uses the minimum weight perfect matching.
For TTP-3, we will select the better one between these two cycles.

\begin{lemma}\label{Second_Hamilton}
Let $M_H$ be a perfect matching in graph $H$. There is a polynomial-time algorithm that can generate a Hamiltonian cycle $C_H$ in $H$ with a weight of at most $ w(M_H)+\frac{1}{2n_S}\Delta_H+O(1/n^2_S)\Delta_H$.
\end{lemma}
\begin{proof}
The perfect matching $M_H$ contains $n_S/2$ edges. We can list the edges of $M_H$ in a random way, and connect them from first to last to form a Hamiltonian cycle $C_H$ in graph $H$, respectively.
For example, suppose $n_S=4$, the listed two edges (in a random way) are denoted by $(v_1, v_2)$ and $(v_3, v_4)$, and then we can construct a Hamiltonian cycle $(v_1,v_2,v_3,v_4,v_1)$.

Recall that the edge set of graph $H$ is denoted by $E_H$.
There are $n_S/2$ newly added edges, belonging to the edge set $E_H\setminus M_H$. Every edge in the set is selected with the same probability $\frac{1}{\size{E_H\setminus M_H}}=\frac{2}{n_S(n_S-2)}$. The expected weight of newly added edges is $(\frac{n_S}{2})\cdot\frac{2}{n_S(n_S-2)}w(E_H\setminus M_H)\leq \frac{1}{2(n_S-2)}\Delta_H=\frac{1}{2n_S}\Delta_H+\frac{1}{n_S(n_S-2)}\Delta_H\leq\frac{1}{2n_S}\Delta_H+\frac{2}{n^2_S}\Delta_H$, where the first inequality follows from $w(E_H\setminus M_H)\leq w(E_H)=\frac{1}{2}\Delta_H$ by (\ref{eqn_GH}), and the second follows from $n_S=n-n_{\overline{S}}\geq 8k^2-4k>4$. Hence, the Hamiltonian cycle $C_H$ has an expected weight of at most
\[
w(M_H)+\frac{1}{2n_S}\Delta_H+\frac{2}{n^2_S}\Delta_H=w(M_H)+\frac{1}{2n_S}\Delta_H+O(1/n^2_S)\Delta_H.
\]

The randomized algorithm runs in $O(n^3)$ time.
Again, it can be derandomized efficiently by the method of conditional expectations \cite{motwani1995randomized}.
\end{proof}

Note that Lemma \ref{Second_Hamilton} holds for any perfect matching. For TTP-3, we consider the Hamiltonian cycle using a minimum weight perfect matching. By making a trade-off between the two Hamiltonian cycles, we can get the following theorem.

\begin{theorem}\label{first_approach}
For TTP-3 with any constant $\epsilon>0$, there is a polynomial-time algorithm with an approximation ratio of $\min\{\frac{4}{3}+\frac{1}{3}\alpha\gamma,\   1-\frac{1}{6}\gamma+\alpha\gamma\}+\epsilon$.
\end{theorem}
\begin{proof}
Here we use $C_H$ to denote a minimum weight Hamiltonian cycle in graph $H$.
If we use the Hamiltonian cycle $C'_H$ obtained by the Christofides-Serdyukov algorithm, we can construct a feasible schedule with a total weight of at most $\frac{2}{3}n_Sw(C'_H)+\frac{2}{3}\Delta_H+O(1/n)\psi$ by Lemma~\ref{lb-distance+}, since $n_Sw(C'_H)\leq \frac{3}{2}n_Sw(C_H)\leq\frac{3}{2}\psi= O(1)\psi$ (recall that the Christofides-Serdyukov algorithm is a $3/2$-approximation algorithm for TSP).
Then, we have that
\begin{align*}
&\frac{2}{3}n_Sw(C'_H)+\frac{2}{3}\Delta_H+O(1/n)\psi\\
&\leq\frac{2}{3}n_S\lrA{\mbox{MST}(H)+\frac{1}{2}w(C_H)}+\frac{2}{3}\Delta_H+O(1/n)\psi\\
&=\frac{2}{3}\lrA{n_S\mbox{MST}(H)+\frac{1}{2}\Delta_H}+\frac{1}{3}\Delta_H+\frac{1}{3}n_Sw(C_H)+O(1/n)\psi\\
&\leq\frac{2}{3}\lrA{1+\frac{1}{4}\gamma}\psi+\frac{1}{3}\lrA{1-\frac{1}{2}\gamma+\alpha\gamma}\psi+\frac{1}{3}\psi+O(1/n)\psi\\
&=\lrA{\frac{4}{3}+\frac{1}{3}\alpha\gamma}\psi+O(1/n)\psi,
\end{align*}
where the first inequality follows from Lemma~\ref{lb-chris}, and the second follows from Lemmas \ref{lb-tsp+}, \ref{lb-delta}, and (\ref{lb-delta_tau}).

Similarly, if we use a minimum weight perfect matching $M_H$ in graph $H$ to obtain the Hamiltonian cycle $C''_H$ in Lemma~\ref{Second_Hamilton}, we can construct a feasible schedule with a total weight of at most $\frac{2}{3}n_Sw(C''_H)+\frac{2}{3}\Delta_H+O(1/n)\psi$, since $n_Sw(C''_H)\leq
n_S(w(M_H)+\frac{1}{2n_S}\Delta_H+O(1/n^2_S)\Delta_H)=n_Sw(M_H)+O(1)\Delta_H= O(1)\psi$ by Lemmas \ref{lb-delta}, \ref{lb-matching} and \ref{Second_Hamilton}. Then, we have that
\begin{align*}
&\frac{2}{3}nw(C''_H)+\frac{2}{3}\Delta_H+O(1/n)\psi\\
&\leq\frac{2}{3}n_S\lrA{w(M_H)+\frac{1}{2n_S}\Delta_H+O(1/n^2_S)\Delta_H}+\frac{2}{3}\Delta_H+O(1/n)\psi\\
&\leq\frac{2}{3}(n_Sw(M)+\Delta_H)+\frac{1}{3}\Delta_H+O(1/n)\psi\\
&\leq\frac{2}{3}(1+\alpha\gamma)\psi+\frac{1}{3}\lrA{1-\frac{1}{2}\gamma+\alpha\gamma}\psi+O(1/n)\psi\\
&=\lrA{1-\frac{1}{6}\gamma+\alpha\gamma}\psi+O(1/n)\psi,
\end{align*}
where the first inequality follows from Lemma~\ref{Second_Hamilton}, the second follows from Lemma~\ref{lb-delta+}, and the last follows from Lemma \ref{lb-delta} and (\ref{lb-delta_matching}).

Since we select the better one between these two Hamiltonian cycles, the approximation ratio is $\min\{\frac{4}{3}+\frac{1}{3}\alpha\gamma,\  1-\frac{1}{6}\gamma+\alpha\gamma\}+O(1/n)$.
Hence, there exists a constant $c$ such that the ratio is bounded by $\min\{\frac{4}{3}+\frac{1}{3}\alpha\gamma,\  1-\frac{1}{6}\gamma+\alpha\gamma\}+c/n$. For any constant $\epsilon>0$, if $n\leq c/\epsilon=O(1)$, we can find an optimal solution by brute force, otherwise we use the approximation algorithm. This establishes the approximation ratio of $\min\{\frac{4}{3}+\frac{1}{3}\alpha\gamma,\  1-\frac{1}{6}\gamma+\alpha\gamma\}+\epsilon$.
\end{proof}

Note that $\max_{0\leq \alpha,\gamma\leq1}\min\{\frac{4}{3}+\frac{1}{3}\alpha\gamma,\  1-\frac{1}{6}\gamma+\alpha\gamma\}+\epsilon$ is maximized when $\alpha=\gamma=1$ with value $(5/3+\epsilon)$. However, this would require $0=(\frac{2}{3}+\frac{1}{3}\gamma-\alpha\gamma)\psi\geq nw(\P^*_3)\geq \frac{1}{2}nw(\C^*_3)$ by Lemma~\ref{lb-packing}. If we use any constant ratio approximation algorithm for minimum weight 3-path packing or 3-cycle packing, we may get a much better schedule based on them. Therefore, we will show that we can do better than $(5/3+\epsilon)$ by combining the Hamiltonian cycle construction with the 3-cycle packing construction shown next.

\subsubsection{The Approximation Quality of the 3-Cycle Packing Construction}
\begin{theorem}\label{second_approach}
For TTP-3 with any constant $\epsilon>0$, if there exist $\rho_c$ and $\rho_p$ approximation algorithms for minimum $3$-cycle and 3-path packing problems respectively, there is a polynomial-time algorithm with an approximation ratio of $(\frac{8\rho+6}{9}+\frac{4\rho-3}{9}\gamma-\frac{4\rho-2}{3}\alpha\gamma+\epsilon)$, where $\rho=\min\{\rho_c, \rho_p\}$.
\end{theorem}
\begin{proof}
First, we consider that the schedule uses a $\rho_c$-approximation 3-cycle packing $\C_3$ in graph $H$. By Theorem \ref{cycle-packing}, we can get a schedule with a weight of at most
\begin{align*}
&\frac{2}{3}\rho_c n_Sw(\C^*_3)+\frac{2}{3}\Delta_H+O(1/n)\psi\\
&\le\frac{4}{3}\rho_c n_Sw(\P^*_3)+\frac{2}{3}\Delta_H+O(1/n)\psi\\
&\le\frac{4}{3}\rho_c\lrA{\frac{2}{3}+\frac{1}{3}\gamma-\alpha\gamma}\psi+\frac{2}{3}\lrA{1-\frac{1}{2}\gamma+\alpha\gamma}\psi+O(1/n)\psi\\
&=\lrA{\frac{8\rho_c+6}{9}+\frac{4\rho_c-3}{9}\gamma-\frac{4\rho_c-2}{3}\alpha\gamma}\psi+O(1/n)\psi,
\end{align*}
where the first inequality follows from $w(\C^*_3)\leq 2w(\P^*_3)$ by the triangle inequality, and the second follows from Lemmas \ref{lb-delta} and \ref{lb-packing}.

Then, we consider that the schedule uses a $\rho_p$-approximation 3-path packing $\P_3$ in graph $H$. By Theorem \ref{path-packing}, we can get a schedule with a weight of at most
\begin{align*}
&\frac{4}{3}\rho_p n_Sw(\P^*_3)+\frac{2}{3}\Delta_H+O(1/n)\psi\\
&\le\frac{4}{3}\rho_p\lrA{\frac{2}{3}+\frac{1}{3}\gamma-\alpha\gamma}\psi+\frac{2}{3}\lrA{1-\frac{1}{2}\gamma+\alpha\gamma}\psi+O(1/n)\psi\\
&=\lrA{\frac{8\rho_p+6}{9}+\frac{4\rho_p-3}{9}\gamma-\frac{4\rho_p-2}{3}\alpha\gamma}\psi+O(1/n)\psi,
\end{align*}
where the inequality follows from Lemmas \ref{lb-delta} and \ref{lb-packing}.

By selecting the better one, we can get an approximation ratio of $(\frac{8\rho+6}{9}+\frac{4\rho-3}{9}\gamma-\frac{4\rho-2}{3}\alpha\gamma+\epsilon)$.
\end{proof}

\subsubsection{Trade-off between Two Constructions}\label{A.4}
\begin{theorem}\label{main-res}
Given a polynomial-time $\rho_c$-approximation algorithm for the minimum $3$-cycle packing problem, and a polynomial-time $\rho_p$-approximation algorithm for the minimum $3$-path packing problem, let $\rho=\min\{\rho_c, \rho_p\}$.
For TTP-3 with any constant $\epsilon>0$, there is a  polynomial-time algorithm with an approximation ratio of
\[
\left\{
\begin{aligned}
\frac{11}{6}-\frac{5}{2(4\rho+1)}+\epsilon,\quad\mbox{if}\quad \rho\leq\frac{9}{4};\\
\frac{5}{3}-\frac{2}{3(4\rho-1)}+\epsilon,\quad\mbox{if}\quad \rho>\frac{9}{4}.
\end{aligned}
\right.
\]
%
\end{theorem}
\begin{proof}
Our algorithm will select the better schedule from the previous two constructions.
By Theorem \ref{first_approach}, the ratio of the Hamiltonian cycle construction is $\min\{\frac{4}{3}+\frac{1}{3}\alpha\gamma,\  1-\frac{1}{6}\gamma+\alpha\gamma\}+\epsilon$. By Theorem \ref{second_approach}, the ratio of the $k$-cycle packing construction is $(\frac{8\rho+6}{9}+\frac{4\rho-3}{9}\gamma-\frac{4\rho-2}{3}\alpha\gamma+\epsilon)$. Therefore, the ratio of our algorithm in the worst case is
\[
\max_{0\leq \alpha,\gamma\leq 1}\min\lrC{\frac{4}{3}+\frac{1}{3}\alpha\gamma,\  1-\frac{1}{6}\gamma+\alpha\gamma,\  \frac{8\rho+6}{9}+\frac{4\rho-3}{9}\gamma-\frac{4\rho-2}{3}\alpha\gamma}+\epsilon.
\]

We can transform it into the following LP, where we let $\widetilde{\gamma}=\alpha\gamma$.
\begin{alignat}{2}
\max\quad & y\nonumber \\
\mbox{s.t.}\quad&y \leq \frac{4}{3}+\frac{1}{3}\widetilde{\gamma},\nonumber \\
&y\leq 1-\frac{1}{6}\gamma+\widetilde{\gamma},\nonumber \\
&y\leq\frac{8\rho+6}{9}+\frac{4\rho-3}{9}\gamma-\frac{4\rho-2}{3}\widetilde{\gamma},\nonumber\\
&0\leq \widetilde{\gamma}\leq\gamma\leq 1.\nonumber
\end{alignat}

The LP shows that the ratio is $\frac{11}{6}-\frac{5}{2(4\rho+1)}+\epsilon$ when $\rho\leq 9/4$, and $\frac{5}{3}-\frac{2}{3(4\rho-1)}+\epsilon$ otherwise.
\end{proof}

There are polynomial-time $8/3$-approximation  algorithms for both the minimum 3-cycle and 3-path packing problems~\cite{DBLP:journals/siamcomp/GoemansW95}. Therefore, we have the following result.

\begin{corollary}
For any constant $\epsilon>0$, there is a polynomial-time $(139/87+\epsilon)$-approximation algorithm for TTP-3.
\end{corollary}

\begin{corollary}\label{ttp3-coro}
For any constant $\epsilon>0$, TTP-3 allows
a polynomial-time $(4/3+\epsilon)$-approximation algorithm if
a minimum 3-cycle or a minimum 3-path packing of $H$ is given.
\end{corollary}

\subsection{TTP-4}\label{C.2}
It is natural to use the same idea to solve TTP-4. For the problem of minimum 4-path packing, to our best knowledge, the current best-known ratio is $3/2$~\cite{DBLP:journals/jda/MonnotT08}. For the problem of minimum 4-cycle packing, there is a $3$-approximation algorithm~\cite{DBLP:journals/siamcomp/GoemansW95}.

If we use the construction based on the Hamiltonian cycle and the construction based on the 4-path packing, we can get a $(17/10+\epsilon)$-approximation algorithm for TTP-4, which improves the previous ratio $(7/4+\epsilon)$~\cite{yamaguchi2009improved}.


\begin{theorem}\label{ttp4-first-approach}
For TTP-4 with any constant $\epsilon>0$, there is a polynomial-time algorithm using the Hamiltonian cycle construction that can achieve an approximation ratio of $(\frac{3}{2}+\frac{1}{8}\gamma+\frac{1}{8}\alpha\gamma+\epsilon)$.
\end{theorem}
\begin{proof}
Here we use $C_H$ to denote a minimum weight Hamiltonian cycle in graph $H$.
If we use the Hamiltonian cycle $C'_H$ obtained by the Christofides-Serdyukov algorithm, we can construct a feasible schedule with a total weight of at most $\frac{3}{4}n_Sw(C'_H)+\frac{1}{2}\Delta_H+O(1/n)\psi$ by Lemma~\ref{lb-distance+}, since $n_Sw(C'_H)\leq \frac{3}{2}n_Sw(C_H)\leq\frac{3}{2}\psi= O(1)\psi$ (recall that the Christofides-Serdyukov algorithm is a $3/2$-approximation algorithm for TSP).
Then, we have that
\begin{align*}
&\frac{3}{4}n_Sw(C'_H)+\frac{1}{2}\Delta_H+O(1/n)\psi\\
&\leq\frac{3}{4}n_S\lrA{\mbox{MST}(H)+\frac{1}{2}w(C_H)}+\frac{1}{2}\Delta_H+O(1/n)\psi\\
&=\frac{3}{4}\lrA{n_S\mbox{MST}(H)+\frac{1}{2}\Delta_H}+\frac{1}{8}\Delta_H+\frac{3}{8}n_Sw(C_H)+O(1/n)\psi\\
&\leq\frac{3}{4}\lrA{\frac{5}{4}+\frac{1}{4}\gamma}\psi+\frac{1}{8}\lrA{\frac{3}{2}-\frac{1}{2}\gamma+\alpha\gamma}\psi+\frac{3}{8}\psi+O(1/n)\psi\\
&=\lrA{\frac{3}{2}+\frac{1}{8}\gamma+\frac{1}{8}\alpha\gamma}\psi+O(1/n)\psi,
\end{align*}
where the first inequality follows from Lemma~\ref{lb-chris}, and the second follows from Lemmas \ref{lb-tsp+}, \ref{lb-delta}, and (\ref{lb-delta_tau}).
Note that there is no need to use the Hamiltonian cycle based on the minimum weight perfect matching.
\end{proof}

\begin{theorem}\label{ttp4-second-approach}
For TTP-4 with any constant $\epsilon>0$, there is a polynomial-time algorithm using the 4-path packing construction that can achieve an approximation ratio of $(\frac{39}{16}+\frac{5}{16}\gamma-\frac{7}{4}\alpha\gamma+\epsilon)$.
\end{theorem}
\begin{proof}
Using the $3/2$-approximation 4-path packing $\P_4$ in graph $H$, by Corollary \ref{path-packing}, we can get a schedule with a weight of at most $\frac{9}{4}n_Sw(\P^*_4)+\frac{1}{2}\Delta_H+O(1/n)\psi$. Then, we have that
\begin{align*}
&\frac{9}{4} n_Sw(\P^*_4)+\frac{1}{2}\Delta_H+O(1/n)\psi\\
&\le\frac{9}{4}\lrA{\frac{3}{4}+\frac{1}{4}\gamma-\alpha\gamma}\psi+\frac{1}{2}\lrA{\frac{3}{2}-\frac{1}{2}\gamma+\alpha\gamma}\psi+O(1/n)\psi\\
&=\lrA{\frac{39}{16}+\frac{5}{16}\gamma-\frac{7}{4}\alpha\gamma}\psi+O(1/n)\psi,
\end{align*}
where the inequality follows from Lemmas \ref{lb-delta} and \ref{lb-packing}.
\end{proof}

Making a trade-off between Theorems \ref{ttp4-first-approach} and \ref{ttp4-second-approach}, we can get an approximation ratio of
\[
\max_{0\leq \alpha,\gamma\leq 1}\min\lrC{\frac{3}{2}+\frac{1}{8}\gamma+\frac{1}{8}\alpha\gamma,\  \frac{39}{16}+\frac{5}{16}\gamma-\frac{7}{4}\alpha\gamma}+\epsilon=\frac{17}{10}+\epsilon.
\]
Hence, we have the following theorem.
\begin{theorem}
For any contant $\epsilon>0$, there is a polynomial-time $(17/10+\epsilon)$-approximation  algorithm for TTP-4.
\end{theorem}

\subsection{TTP-$k$ with $k\geq 5$}\label{C.3}
For TTP-$k$ with $k\geq5$, we note that both of the best-known approximation ratios for minimum $k$-cycle and $k$-path packing problems are $4(1-1/k)>3$~\cite{DBLP:journals/siamcomp/GoemansW95}. The approximation ratios are too big and we cannot improve TTP-$k$ by using the same idea. However, by a refined analysis in \cite{yamaguchi2009improved}, we can slightly improve the approximation ratio $(\frac{5k-7}{2k}+\epsilon)$ by a term of $\Theta(1/k)$.

Define $\zeta=\alpha\gamma+\beta(1-\gamma)$. It measures the proportion of weights of all home-edges compared to the itineraries of all teams. Then, we have $0\leq \zeta\leq 1$. We can get two following lower bounds.

\begin{lemma}\label{lb-delta++}
$(\frac{k-2}{2}+\zeta)\psi\geq \Delta_H$.
\end{lemma}
\begin{proof}
By Lemma~\ref{lb-delta}, we have that $\Delta_H\leq (\frac{k-2}{2}+\alpha)\gamma\psi +(\frac{k-3}{2}+\beta)(1-\gamma)\psi$. Note that $(\frac{k-2}{2}+\alpha)\gamma\psi +(\frac{k-3}{2}+\beta)(1-\gamma)\psi\leq (\frac{k-2}{2}+\alpha)\gamma\psi +(\frac{k-2}{2}+\beta)(1-\gamma)\psi=(\frac{k-2}{2}+\zeta)\psi$.
\end{proof}

\begin{lemma}\label{lb-tree+}
$\min\{(1-\frac{1}{2}\zeta)\psi, \frac{k}{k+1}\psi\} \geq n_S\mbox{MST}(H)$.
\end{lemma}
\begin{proof}
We first show that $(1-\frac{1}{2}\zeta)\psi\geq n_S\mbox{MST}(H)$.
By Lemma~\ref{lb-tree}, we have that $n_S\mbox{MST}(H)\leq (1-\frac{1}{2}\alpha)\gamma\psi+(1-\frac{1}{2}\beta)(1-\gamma)\psi=(1-\frac{1}{2}\zeta)\psi$.

To prove $\frac{k}{k+1}\psi\geq n_S\mbox{MST}(H)$, it is sufficient to prove that $\frac{k}{k+1}w(I_v)\geq \mbox{MST}(H)$ holds for any $v\in S$. Recall that $I_v$ contains a set of $k'$-cycles with $3\leq k'\leq k+1$. For each $k'$-cycle $C$, we delete the longest edge and the weight of the cycle loses at least $\frac{1}{k'}w(C)\geq \frac{1}{k+1}w(C)$. Hence, we can get a spanning tree with the a total weight of at most $\frac{k}{k+1}w(I_v)$. We have $\frac{k}{k+1}w(I_v)\geq \MST(H)$.
\end{proof}

By Lemmas \ref{lb-delta++} and \ref{lb-tree+}, we have that
\begin{equation}\label{lb-delta_tau+}
\frac{k+2}{4}\psi\geq \frac{1}{2}\Delta_H+n_S\mbox{MST}(H).
\end{equation}

\begin{theorem}\label{ttp-k}
For TTP-$k$ with any constant $k\geq5$,
there is a polynomial-time algorithm achieving  approximation ratio of $(\frac{5k^2-4k+3}{2k(k+1)}+\epsilon)$.
\end{theorem}
\begin{proof}
Here we use $C_H$ to denote a minimum weight Hamiltonian cycle in graph $H$.
If we use the Hamiltonian cycle $C'_H$ obtained by the Christofides-Serdyukov algorithm, we can construct a feasible schedule with a total weight of at most $\frac{k-1}{k}n_Sw(C'_H)+\frac{2}{k}\Delta_H+O(k^2/n)\psi$ by Lemma~\ref{lb-distance+}, since $n_Sw(C'_H)\leq \frac{3}{2}n_Sw(C_H)\leq\frac{3}{2}\psi= O(1)\psi$ (recall that the Christofides-Serdyukov algorithm is a $3/2$-approximation algorithm for TSP).
Then, we have that
\begin{align*}
&\frac{k-1}{k}n_Sw(C'_H)+\frac{2}{k}\Delta_H\\
&\leq\frac{k-1}{k}n_S\lrA{\mbox{MST}(H)+\frac{1}{2}w(C_H)}+\frac{2}{k}\Delta_H\\
&=\frac{4}{k}\lrA{n_S\mbox{MST}(H)+\frac{1}{2}\Delta_H}+\frac{k-5}{k}n_S\mbox{MST}(H)+\frac{k-1}{2k}n_Sw(C_H)\\
&\leq\frac{4}{k}\lrA{\frac{k+2}{4}\psi}+\frac{k-5}{k}\lrA{\frac{k}{k+1}\psi}+\frac{k-1}{2k}\psi\\
&=\frac{5k^2-4k+3}{2k(k+1)}\psi,
\end{align*}
where the first inequality follows from Lemma~\ref{lb-chris}, and the second follows from Lemmas \ref{lb-tsp+}, \ref{lb-tree+}, and (\ref{lb-delta_tau+}).
Since we consider $k$ as a constant,
we can get an approximation of $(\frac{5k^2-4k+3}{2k(k+1)}+\epsilon)$.
Note that it satisfies that $\frac{5k^2-4k+3}{2k(k+1)}\leq\frac{5k-7}{2k}$ for any constant $k\geq5$.
\end{proof}

\section{LDTTP-$k$ with $k\geq 3$}\label{D}
If all teams are on a line shape graph, TTP-$k$ is known as Linear Distance TTP-$k$ (LDTTP-$k$)~\cite{DBLP:journals/jair/HoshinoK12}. For this graph, it is easy to see that the minimum $k$-cycle/path packing can be found in polynomial time.
For LDTTP-3, by Corollary~\ref{ttp3-coro}, our algorithm has an approximation ratio of $(4/3+\epsilon)$, which also matches the current best-known ratio $4/3$~\cite{DBLP:journals/jair/HoshinoK12}.
Their result is based on a stronger lower bound of LDTTP-3.
If we use this lower bound, the ratio of our 3-cycle packing construction can achieve $(6/5+\epsilon)$. In this section, we will consider LDTTP-$k$ directly.

\subsection{The Lower Bound}
We first propose a lower bound for LDTTP-$k$, which is a generalization of the lower bound for LDTTP-3 used in~\cite{DBLP:journals/jair/HoshinoK12}.
We consider the graph $H$, which is also a line shape.
Assume that the teams from left to right on the line are $(t_1,t_2,\dots, t_{n_S})$, respectively. We let $d_i=w(t_i,t_{i+1})$ for all $1\leq i\leq n_S-1$. Then, we have the following lemma.
\begin{lemma}\label{lb-LDTTP}
For LDTTP-$k$, it holds that
$
\psi\geq\sum_{i=1}^{n_S-1}c_id_i,
$
where $c_i=2i\lceil\frac{n_S-i}{k}\rceil+2(n_S-i)\lceil\frac{i}{k}\rceil$.
\end{lemma}
\begin{proof}
For an edge $t_it_{i+1}$, there are $i$ teams on the left and $n_S-i$ teams on the right of it. For each left team, it takes at least $\ceil{\frac{n_S-i}{k}}$ trips for it to visit all right teams. For each right team, it takes at least $\ceil{\frac{i}{k}}$ trips for it to visit all left teams.
There are at least $i\ceil{\frac{n_S-i}{k}}+(n_S-i)\ceil{\frac{i}{k}}=\frac{l_i}{2}$ trips in total. Each trip must cross the edge twice and hence all teams must cross the edge at least $l_i$ times. Together, we get the lower bound $\sum_{i=1}^{n_S-1}l_id_i$.
\end{proof}

Next, we analyze the approximation ratio of our $k$-cycle packing schedule.

\subsection{The Analysis}
We can get an optimal $k$-cycle packing $\C_k$ in graph $H$ such that
\[
\C_k=\{(t_{ki-(k-1)},t_{ki-(k-2)},\dots,t_{ki},t_{ki-(k-1)})\}_{i=1}^{n_S/k}.
\]
Using $\C_k$, by Theorem \ref{cycle-packing}, the weight of our cycle packing schedule is $\frac{k-1}{k}n_Sw(\C_k)+\frac{2}{k}\Delta_H+O(k^3/n)\psi$.
Let $\frac{k-1}{k}n_Sw(\C_k)+\frac{2}{k}\Delta_H=\sum_{i=1}^{n_S-1}f_id_i$. Define $\rho=\max_{1\leq i\leq n_S-1}\frac{f_i}{c_i}$. Then, by Lemma~\ref{lb-LDTTP}, we have
\begin{equation}\label{lb_LDTTP_k}
\begin{aligned}
\frac{k-1}{k}n_Sw(\C_k)+\frac{2}{k}\Delta_H=\sum_{i=1}^{n_S-1}f_id_i
\leq\rho\sum_{i=1}^{n_S-1}c_id_i\leq\rho\psi.
\end{aligned}
\end{equation}

\begin{theorem}\label{ldttp-final}
For LDTTP-$k$ with any constant $k\geq3$, the approximation ratio of our cycle packing schedule is $(\frac{3k-3}{2k-1}+\epsilon)$.
\end{theorem}
\begin{proof}
By (\ref{lb_LDTTP_k}), we know that it is sufficient to prove $\rho= \frac{3k-3}{2k-1}+O(1/n)$.
Recall that $\frac{k-1}{k}n_Sw(\C_k)+\frac{2}{k}\Delta_H=\sum_{i=1}^{n_S-1}f_id_i$.
On one hand, it is easy to see that $\frac{k-1}{k}n_Sw(\C_k)=\sum_{i=1}^{n_S/k}\sum_{j=1}^{k-1}\frac{2k-2}{k}n_Sd_{ki-j}$. On the other hand, we have $\frac{2}{k}\Delta_H=\sum_{i=1}^{n_S-1}\frac{4}{k}i(n-i)d_i$.
Therefore, we have $f_i=\frac{4}{k}i(n_S-i)$ for $i\bmod k=0$, and $f_i=\frac{4}{k}i(n_S-i)+\frac{2k-2}{k}n_S$ for $i\bmod k\neq0$.
Recall that $c_i=2i\lceil\frac{n_S-i}{k}\rceil+2(n_S-i)\lceil\frac{i}{k}\rceil$. We have $f_{i'}=f_{n_S-i'}$ and $c_{i'}=c_{n_S-i'}$ for any $1\leq i'\leq n_S/2$. The structure has a symmetry property, and hence we only consider the case of $1\leq i\leq n_S/2$.

\textbf{Case 1: $i\bmod k=0$.} We have $f_i=\frac{4}{k}i(n_S-i)$ and $c_i=\frac{4}{k}i(n_S-i)$. Then, we can get $\frac{f_{i}}{c_{i}}=1$.

\textbf{Case 2: $i\bmod k=j$, where $1\leq j<k$.} We have $f_i=\frac{4}{k}i(n_S-i)+\frac{2k-2}{k}n_S$.
Note that $\ceil{\frac{n_S-i}{k}}=\frac{n_S-i+j}{k}$ and $\ceil{\frac{i}{k}}=\frac{i+k-j}{k}$.
We have $c_i=2i\lceil\frac{n_S-i}{k}\rceil+2(n_S-i)\lceil\frac{i}{k}\rceil=\frac{4}{k}i(n_S-i)+\frac{2}{k}((k-j)n_S+(2j-k)i)$. Then, we can get
\begin{align*}
\frac{f_{i}}{c_{i}}=\ &\frac{\frac{4}{k}i(n_S-i)+\frac{2k-2}{k}n_S}{\frac{4}{k}i(n_S-i)+\frac{2}{k}((k-j)n_S+(2j-k)i)}\\
=\ &1+\frac{\frac{2}{k}((j-1)n_S-(2j-k)i)}{\frac{4}{k}i(n_S-i)+\frac{2}{k}((k-j)n_S+(2j-k)i)}\\
=\ &1+\frac{(j-1)n_S-(2j-k)i}{2i(n_S-i)+(k-j)n_S+(2j-k)i}.
\end{align*}
In the worst case, it is easy to see that we have $j>1$.
Recall that $1\leq i\leq n_S/2$ and $k=\Theta(1)$. If $i=\Theta(n)$, we have $(j-1)n_S-(2j-k)i=\Theta(n)$, $2i(n_S-i)+(k-j)n_S+(2j-k)i=\Theta(n^2)$, and hence $f_i/c_i=1+\Theta(1/n)$. Otherwise, if $i=o(n)$ and $i=\omega(1)$, we have $(j-1)n_S-(2j-k)i=\Theta(n)$, $2i(n_S-i)+(k-j)n_S+(2j-k)i=\omega(n)$, and hence $f_i/c_i=1+o(1)$. Hence, for the rest case, we have $i=\Theta(1)$. Moreover, we have that $(j-1)n_S-(2j-k)i=(j-1)n_S+\Theta(1)$ and $2i(n_S-i)+(k-j)n_S+(2j-k)i=(k+2i-j)n_S+\Theta(1)$. Therefore, in the worst case, we have that
\begin{align*}
\frac{f_{i}}{c_{i}}=\ &1+\frac{(j-1)n_S+\Theta(1)}{(k+2i-j)n_S+\Theta(1)}\\
=\ &1+\frac{j-1}{k+2i-j}+\Theta(1/n)\\
\leq\ &1+\frac{j-1}{k+j}+\Theta(1/n)\\
\leq\ &1+\frac{k-2}{2k-1}+\Theta(1/n),
\end{align*}
where the first inequality follows from $i\geq j$ since $i\bmod k=j$ and the second is from $1\leq j<k$.

By the above analysis, we have $\rho=\max_{1\leq i\leq n_S-1}\frac{f_i}{c_i}=1+\frac{k-2}{2k-1}+O(1/n)=\frac{3k-3}{2k-1}+O(1/n)$.
\end{proof}

We can see that the approximation ratio is $(6/5+\epsilon)$ for LDTTP-3, which improves the previous result of $(4/3+\epsilon)$~\cite{DBLP:journals/jair/HoshinoK12}.

For LDTTP-$k$, we can also analyze the Hamiltonian cycle schedule. Let $C_H$ be the optimal Hamiltonian cycle in graph $H$, i.e., $C_H=(t_1,\dots, t_{n_S},t_1)$. Note that $n_Sw(C_H)\leq \psi$ by Lemma~\ref{lb-tsp+}. By Lemma~\ref{lb-distance+}, the total weight of the Hamiltonian cycle schedule is $\frac{2k-2}{k}n_Sw(C_H)+\frac{2}{k}\Delta_H+O(k^2/n)\psi$.
Let $\frac{2k-2}{k}n_Sw(C_H)+\frac{2}{k}\Delta_H=\sum_{i=1}^{n-1}f'_id_i$.
We have that $\frac{2k-2}{k}n_Sw(C)=\sum_{i=1}^{n_S-1}\frac{2k-2}{k}n_Sd_i$, $\frac{2}{k}\Delta_H=\sum_{i=1}^{n_S-1}\frac{4}{k}i(n_S-i)d_i$, and $c_i=2i\lceil\frac{n_S-i}{k}\rceil+2(n_S-i)\lceil\frac{i}{k}\rceil$. It is easy to see that the bottleneck case is
\[
\frac{f'_k}{c_k}=\frac{\frac{2k-2}{k}n_S+4(n_S-k)}{4(n_S-k)}=\frac{3k-1}{2k}+O(1/n).
\]

\begin{theorem}
For LDTTP-$k$ with any constant $k\geq3$, the approximation ratio of the Hamiltonian cycle schedule is $(\frac{3k-1}{2k}+\epsilon)$.
\end{theorem}

We can see that the approximation ratio $(\frac{3k-1}{2k}+\epsilon)$ is worse than the ratio $(\frac{3k-3}{2k-1}+\epsilon)$.
A natural question is whether there is an algorithm with an approximation ratio better than $(\frac{3k-3}{2k-1}+\epsilon)$.

\section{Conclusion}\label{E}
In this paper, we consider approximation algorithms for TTP-$k$. We propose several new lower bounds and a new construction based on the $k$-cycle/path packing. Our construction can be combined well with the previous construction on the Hamiltonian cycle. For TTP-3, which is the most investigated case of TTP-$k$, we prove that these constructions can guarantee an $(139/87+\epsilon)$-approximation ratio which is better than the best-known ratio $(5/3+\epsilon)$. Moreover, we show that we can improve the ratio from $(7/4+\epsilon)$ to $(17/10+\epsilon)$ for TTP-4 and from $(\frac{5k-7}{2k}+\epsilon)$ to $(\frac{5k^2-4k+3}{2k(k+1)}+\epsilon)$ for TTP-$k$ with $k\geq 5$. Furthermore, for LDTTP-$k$, we show that the cycle packing construction has an approximation of $(\frac{3k-3}{2k-1}+\epsilon)$, which improves the ratio from $4/3$ to $(6/5+\epsilon)$ for LDTTP-3.
We also remark that according to our results, an improved approximation ratio for minimum weight $k$-cycle/path packing may directly lead to further improvements for TTP-$k$. 

\backmatter


\bmhead{Acknowledgements}
The work is supported by the National Natural Science Foundation of China, under grant 62372095.

\section*{Declarations}
\subsection*{Funding}
Not applicable.

\subsection*{Competing interests}
The authors declare no competing interests.

\subsection*{Ethics approval and consent to participate}
Not applicable.

\subsection*{Consent for publication}
Not applicable.

\subsection*{Data availability}
Not applicable.

\subsection*{Materials availability}
Not applicable.

\subsection*{Code availability}
Not applicable.

\subsection*{Author contribution}
All authors wrote the main manuscript text. All authors reviewed the manuscript.
\bibliography{mybib}
\end{document}